\documentclass[letterpaper, 10 pt, journal]{IEEEtran}

\IEEEoverridecommandlockouts

\usepackage{amsmath,amssymb,amsfonts}
\usepackage{amsthm}
\usepackage{bbm}
\usepackage{algorithmic}  
\usepackage{graphicx}
\usepackage{textcomp} 
\usepackage{xcolor}
\usepackage{amsmath,amssymb,amsfonts}
\usepackage{algorithmic} 
\usepackage{graphicx}
\usepackage{textcomp}
\usepackage{xcolor}
\usepackage{algorithmic}
\usepackage{algorithm}
\usepackage[nopar]{lipsum}
\allowdisplaybreaks
\newcommand\numberthis{\addtocounter{equation}{1}\tag{\theequation}}

\theoremstyle{plain}
\newtheorem{thm}{Theorem}
\newtheorem{lem}[thm]{Lemma}
\newtheorem{prop}[thm]{Proposition}

\newtheorem{rem}[thm]{Remark}

\theoremstyle{remark}

\newtheorem{remark}{Remark}
 
\theoremstyle{plain}

\theoremstyle{definition}

\theoremstyle{remark}

\DeclareMathOperator*{\argmax}{arg\,max}

\title{\LARGE \bf Learning to Cache and Caching to Learn: Regret Analysis of Caching Algorithms
}

\author{Archana Bura, Desik Rengarajan, Dileep Kalathil, Srinivas Shakkottai, and Jean-Francois Chamberland-Tremblay
\thanks{Authors are with the Department of Electrical and Computer Engineering,  Texas A\&M University, Texas, USA. Email of the corresponding author:  
        {\tt\small  dileep.kalathil@tamu.edu}}%
}

\begin{document}

\maketitle

\begin{abstract}
Crucial performance metrics of a caching algorithm include its ability to quickly and accurately learn a popularity distribution of requests.  However, a majority of work on analytical performance analysis focuses on hit probability after an asymptotically large time has elapsed.  We consider an online learning viewpoint, and characterize the ``regret'' in terms of the finite time difference between the hits achieved by a candidate caching algorithm with respect to a genie-aided scheme that places the most popular items in the cache.   We first consider the Full Observation regime wherein all requests are seen by the cache.  We show that the Least Frequently Used (LFU) algorithm is able to achieve order optimal regret, which is matched by an efficient counting algorithm design that we call LFU-Lite.  We then consider the Partial Observation regime wherein only requests for items currently cached are seen by the cache, making it similar to an online learning problem related to the multi-armed bandit problem.  We show how approaching this ``caching bandit'' using traditional approaches yields either high complexity or regret, but a simple algorithm design that exploits the structure of the distribution can ensure order optimal regret.  We conclude by illustrating our insights using numerical simulations.
\end{abstract}

\section{Introduction}
\label{section: introduction}
Caching is a fundamental aspect of content distribution.  Since it is often the case that the same content item is requested by multiple clients over some timescale, replicating and storing content in near proximity to the requesting clients over that timescale can both reduce latency at the clients, as well as enable more efficient usage of network and server resources.   Indeed, this is the motivation for a  variety of cache eviction policies such as Least Recently Used (LRU), First In First Out (FIFO), RANDOM, CLIMB \cite{coffman73operating}, Least Frequently Used (LFU) \cite{einziger2017tinylfu} etc., all of which attempt to answer a basic question: suppose that you are aware of the timescale of change of popularity, what are the right content items to store?

Viewed from this angle, the problem of caching is simply that there is some underlying unknown popularity distribution (that could change with time) over a library of content items, and the goal of a caching algorithm is to quickly learn which items are most popular and place them in a location that minimizes client latencies. Taking this viewpoint of ``caching equals fast online learning of an unknown probability distribution,''  it is clear that it is not sufficient for a caching algorithm to learn a fixed popularity distribution accurately, \emph{it must also learn it quickly} in order to track the changes on popularity that might happen frequently. 

Most work on the performance analysis of caching algorithms has focused on the stationary (long term) hit probabilities under a fixed request distribution.  However, such an approach does not account for the fact that request distributions change with time, and finite time performance is a crucial metric.  Suppose that all content items are of the same size, a cache can hold $C$ content items, and the request process consists of independent draws (called the Independent Reference Model (IRM)).  Then a genie-aided algorithm that is aware of the underlying popularity distribution would place the top $C$ most popular items in the cache to maximize the hit probability.  Yet, any pragmatic caching algorithm needs to learn the popularity distribution as requests arrive, and determine what to cache.  The \emph{regret} suffered is the difference in the number of cache hits between the two algorithms.  How does the regret scale with the number of requests seen? 

\begin{figure}
\centering
\includegraphics[width=0.9\linewidth]{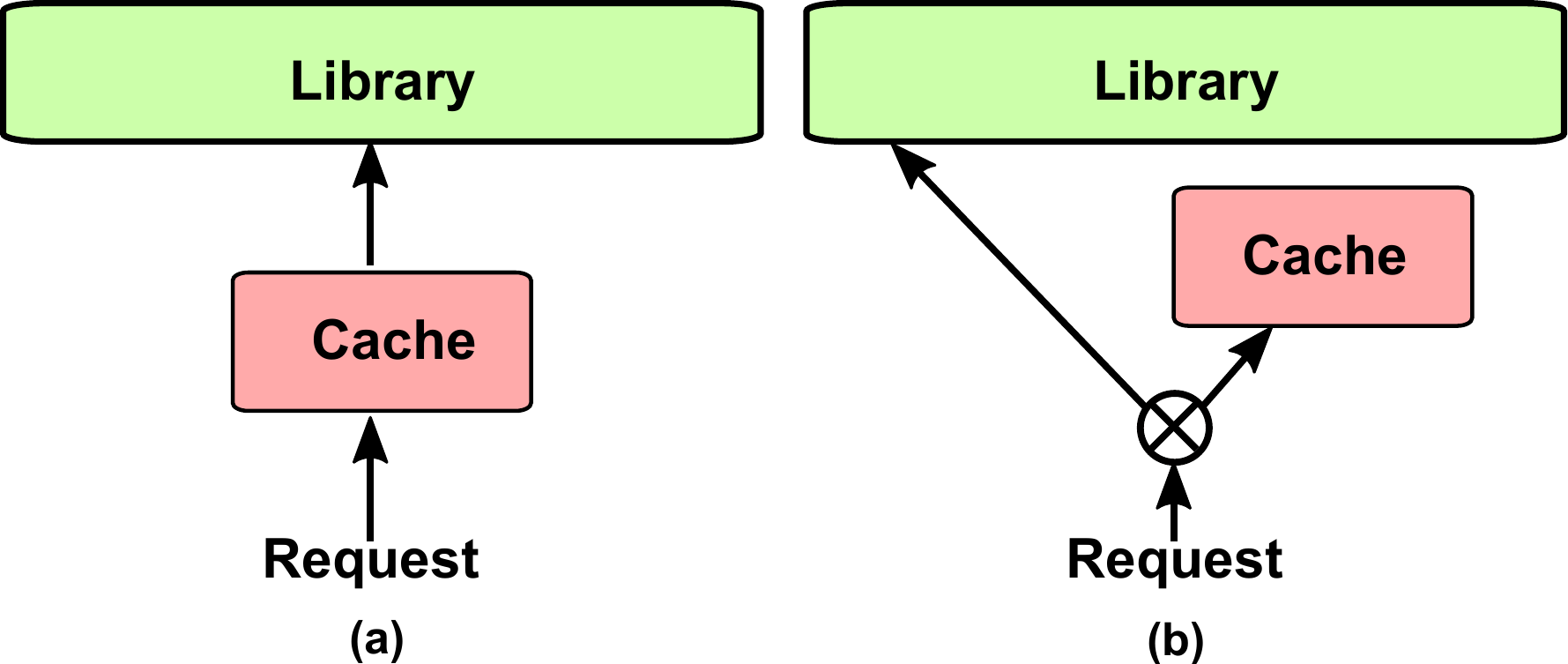}
\caption{{Request forwarding with (a) Full observation, and (b) Partial observation at the cache.}}
\label{fig:full-partial}
\vspace{-0.2in}
\end{figure}
The goal of this work is to conduct a systematic analysis of caching from the perspective of regret, with idea that a low-regret algorithm implies fast and accurate learning in finite time, and hence should be usable in a setting where the popularity distribution changes with time.  We consider two different caching paradigms, illustrated in Figure~\ref{fig:full-partial}. In both paradigms, there is a single cache of size $C$ and a much larger library of items of size $L.$  In the first setup that we call the \emph{full observation} model, all requests are first directed towards the cache, and then forwarded to the library if there is a cache miss.  Here, the cache sees all requests and can learn the underlying popularity distribution with full observation of the requests.  The second model is more akin to a data center with multiple nodes, wherein requests may be forwarded via cache routing towards a node that has the item.  We call this case as the \emph{partial observation} model, since the cache only sees requests for items currently cached in it, while the other requests are directly forwarded to the library. Thus, items must be cached in an exploratory manner in order to learn their popularities, and whether they are worth retaining.  Can we design regret optimal algorithms that apply to each of these cases?

\subsection*{Main Results}

In our analytical model, we consider a system in which one request arrives at each discrete time unit, i.e., the total number of requests is the same as the elapsed time $T$.  We begin with insight that, under the full observation regime, the empirical frequency is a sufficient statistic of all information on the popularity distribution received thus far.  Furthermore, the appropriate use of this estimate is to choose the top $C$ most frequent items to cache. This approach is identical to the LFU, since it evicts the item with the least empirical frequency at each time.   Our first result is to show that LFU has an $O(1)$ regret, not only with respect to time $T$, but also with respect to library size.

While LFU is known to attain high hit rates, it suffers from the fact that the number of counters is the same as the library size, since every request must be counted.  This is clearly prohibitive, and has given rise to approximations such as W-LFU \cite{karakostas2002exploitation}, which only keeps counts within a moving window of requests, and TinyLFU \cite{einziger2017tinylfu} that uses a sketch for approximate counting.  Our next result is to show that these approximations never entirely eliminate the error in estimating the popularity distribution, leading to the worst possible regret of $\Omega (T).$

We then propose a varient of LFU that we term LFU-Lite, under which we use a moving window of requests to decide whether or not a particular item appears to be popular enough to be counted accurately.  Thus, we maintain a counter bank, and only count those content items that meet a threshold frequency in any window of requests thus far.  The counter bank size grows in a concave manner with time, and we find its expected size to ensure $O(1)$ regret for a target time $T.$  Thus, given a time constant of change in popularity, we can decide on the ideal number of counters.

We next consider the partial observation regime, wherein the cache can only see requests of items currently cached in it.  We relate this problem to that of the classical multi-armed bandit under which actions must be taken to learn the value of pulling the different arms.  We first consider an algorithm that builds up the correct posterior probabilities given the requests seen thus far, and caches the most frequent items in a sample of this posterior distribution.  Although its empirical performance is excellent, maintaining the full posterior sampling (FPS) quickly becomes prohibitively difficult.

We then consider an algorithm that simply conducts a marginal posterior sampling (MPS) by updating counts only for the items that are in the cache.  Here, counts of hits and misses are awarded to the appropriate cached item, but a miss (which manifests itself as no request being made to the cache) is not used to update the posterior distribution of items not in the cache.  Clearly, we are not using the information effectively, and this is reflected in the regret scaling as $O(\log T).$  This result is similar to earlier work~\cite{blasco2014learning}.

We then ask whether we can exploit the structure of the problem to do better?  In particular, suppose that we know that requests will follow a certain probability distribution (e.g., Zipf), although we do not know the ranking of items (i.e., we do not know which one is the most popular etc.).  We develop a Structured Information (SI) algorithm that considers this information about the distribution to reduce the regret to $O(1)$.  We also describe a ``Lite'' version of the SI algorithm similar to LFU-Lite to reduce the number of counters.

We first verify our analytical results via numerical simulations conducted using an IRM model drawn from Zipf distributions with different parameters of library and cache sizes.  We also find that Lite-type schemes appear to empirically  perform even better than predicted by the analytical results.  

We then construct versions of the algorithms that are capable of following a changing popularity distribution by simply "forgetting" counts, which takes the form of periodically halving the counts in the counters.  The expectation is that a low regret algorithm, augmented with such a forgetting rule with an appropriately chosen periodicity should be able to track a moving popularity distribution accurately.  We conduct trace-based simulations using (non-stationary) data sets obtained from IBM and YouTube, and compare hit performance against the ubiquitous LRU algorithm.  We show that the LFU variants outperform LRU, and that incorporating forgetting enhances their hit-rates.  

Since the amount of change over time in the existing traces is low, we stress test our algorithms by creating a synthetic trace that has higher changes in popularity over time.  Again, we show that the versions of our algorithms that incorporate forgetting are able to track such changing distributions, and are still able to outperform LRU, which builds a case for their eventual adoption.

\subsection*{Related Work}

Existing analytical studies of caching algorithms largely follow the IRM model, with the focus being on closed form results of the stationary hit probabilities of LRU, FIFO, RANDOM, and CLIMB \cite{king71, coffman73operating, gelenbe73unified,starobinski01}.  The expressions are often hard to compute for large caches, and  approximations have been proposed for larger cache sizes \cite{rosensweig2010}.  Of particular interest is the  Time-To-Live (TTL) approximation \cite{fagin77,che02,berger14exact,gast16asym} that associates each cached item with a lifetime after which it is evicted.  Appropriate choice of this lifetime enables the accurate approximation of different caching schemes \cite{berger14exact}.  

Recent work on performance analysis of caching algorithms has focused on the online learning aspect.  For instance, \cite{basu2018adaptive} propose TTL-based schemes to show that a desired hit rate can be achieved under non-stationary arrivals.  Other work such as \cite{li2018accurate} characterize the mixing times of several simple caching schemes such as LRU, CLIMB, k-LRU etc. with the goal of identifying their learning errors as a function of time.  However, the algorithms studied all have stationary error (they never learn perfectly) and so regret in our context would be $\Omega(T)$.

Information Centric caching has gained much recent interest, and is particularly relevant to edge wireless networks.  Joint caching and routing is studied in \cite{ioannidis2018jointly} where the objective is to show asymptotic accuracy of the placements, rather than finite time performance that we focus on.   Closest to our ideas on the partial observation model is work such as \cite{blasco2014learning}, which draws a parallel between bandit algorithms and caching under this setting.  However, the algorithms considered are in the manner of the traditional Multi-Armed Bandit (MAB) approach that does not account for problem structure, and hence can only attain $O(\log T)$ regret.

{With regard to the MAB problem, Lai and Robbins \cite{lr} showed in seminal work the $\Omega(\log T)$ regret lower bound pertaining to any online learning algorithm.  An index based algorithm using the upper confidence idea (UCB1 algorithm) was proposed in \cite{auer2002finite}, which enabled a simple implementation while achieving the optimal regret.  The the posterior sampling approach, first proposed by Thompson in \cite{thompson1933likelihood}, has recently been shown to attain optimal regret \cite{agrawal2013further}. For a detailed survey, we point to a monograph \cite{bubeck2012regret} and a recent book \cite{cesa2012combinatorial}.}   


Much work also exists on the empirical performance evaluation of caching algorithms using traces gathered from different applications.  While several discover fundamental insights \cite{martina2014,zink08,megiddo03arc}, our goal in this work is on analytical performance guarantees, and we do not provide a comprehensive review.

\section{System Model}
\label{section: System Model} 
We consider the optimal cache content placement problem in a communication network.  The  library, which is the set of all files, is denoted by $\mathcal{L} = \{1, \ldots, L\}$.   We assume for expositional simplicity that all files are of the same size, and that the cache has a capacity of $C$, i.e., it can store $C$ files at a given time. We denote the popularity of the files by the profile $\mu = (\mu_{1}, \mu_{2}, \ldots, \mu_{L}),$ with  $\sum_{i} \mu_{1} =1$.  Without loss of generality, we assume that $\mu_{1} > \mu_{2} > \cdots > \mu_{L}$.  Let $x(t) \in \mathcal{L}$ be the file request received at time $t$.  We assume that requests are generated independently according to the popularity profile $\mu$, i.e., $\mathbb{P}(x(t) = i) = \mu_{i}.$ 

Let $C(t)$ be set of files placed in the cache by the caching algorithm at time $t$. We say that cache gets a \textit{hit} if $x(t) \in C(t)$ and a \textit{miss} if $x(t) \notin C(t)$. The goal of the caching algorithm is to maximize the expected  cumulative hits over time, $\mathbb{E}[\sum^{T}_{t=1} \mathbbm{1}\{x(t) \in C(t)\} ],$ where the expectation is over all the random choices on $C(t)$ made by the caching algorithm, and  the randomness over the requests $x(t)$.   Clearly, if popularity distribution $\mu$ is known, the optimal caching policy is to place the most popular items in the cache all the time, i.e., $C^{*}(t) = \mathcal{C}$, where $\mathcal{C} = \{1, 2, \ldots, C\}$. However, in most real world applications, the popularity distribution is unknown to the caching algorithm apriori.   So the goal of a caching algorithm is to learn the popularity distribution from the sequential observations, and to place files in the cache  by optimally using this  available observation at each time in order to maximize the expected cumulative hits.

In the multi-armed bandits literature, it is common to characterize the performance of an online learning algorithm using the metric of \textit{regret}.  Regret is defined as the performance loss of the  algorithm as compared to the optimal algorithm with complete information.  Since $C^{*}(t) = \mathcal{C}$, the cumulative regret of a caching algorithm after $T$ time steps is defined as
\begin{align}
R(T) = \sum^{T}_{t=1} \mathbbm{1}\{x(t) \in \mathcal{C}\} - \mathbbm{1}\{x(t) \in C(t)\}. 
\end{align}

Let $s(t)$ be the observation available to the caching algorithm at time $t$ and let  $h(t) = (s(1), \ldots, s(t-1))$ be the history of observations until time $t$.  The \textit{optimal caching problem} is defined as the problem of finding a policy $\pi(\cdot)$ that maps $h(t)$ to $C(t)$, i.e., $C(t) = \pi(h(t))$, in order to minimize  the   expected cumulative regret, $\mathbb{E}[R(T)]$.     

The choice of the caching policy will clearly depend on the nature of the sequential observations available to it. We consider two different observation structures that are  most common in communication networks.  
\begin{enumerate}
\item \textbf{Full Observation:}   In the full observation structure, we assume that the caching algorithm is able to observe the file request at each time, i.e., $s(t) = x(t)$.   In the setup of a cache and library, this regime corresponds to all requests being sent to the cache, which can then forward the request to the library in case of a miss.

\item \textbf{Partial Observation:} In the partial observation structure, the caching algorithm can observe the request only in the case of a hit, i.e., only if the requested item is in the cache already. More precisely, we define $s(t) = x(t) \mathbbm{1}\{x(t) \in C(t)\}$ under this observation structure. In the case of a miss, $s(t) = 0.$  In the setup of a cache and library, this regime corresponds to the context of information centric caching, wherein requests are forwarded to the cache only if the corresponding content is cached.
\end{enumerate}

We propose  different caching algorithms to address the optimal caching problem under these two observation structures.  
\section{Caching with Full Observation}

We first consider the full observation structure where the caching algorithm can observe every file request.  
Our focus is on a class of algorithms following Least Frequently Used (LFU) eviction, since it uses cumulative statistics of all received requests (unlike other popular algorithms such as Least Recently Used (LRU)), and so is likely to have low regret.


\subsection{LFU Algorithm}

At each time $t$, the LFU algorithm selects the top $C$ requested files until time $t$ and  places them in the cache.  More precisely, LFU algorithm maintains an empirical estimate of the popularity distribution, denoted by $\hat{\mu}(t) = (\hat{\mu}_{1}(t), \ldots, \hat{\mu}_{L}(t))$,  defined as 
\begin{align*}
\hat{\mu}_{i}(t) = \frac{1}{t} \sum^{t}_{\tau=1} \mathbbm{1}\{x(\tau) = i\}, \quad \forall i \in \mathcal{L}. 
\end{align*}
The files to be placed in cache at time $t+1$, $C_{\text{LFU}}(t+1)$, is then selected as
\begin{align*}
C_{\text{LFU}}(t+1) = \arg \max_{C} ~ (\hat{\mu}_{1}(t), \ldots, \hat{\mu}_{L}(t))
\end{align*}
where $\arg \max_{C}$ indicates the indices of the top $C$ elements of the vector $\hat{\mu}(t)$.

We now present the finite time performance guarantee for the LFU algorithm.
\begin{thm} 
\label{thm1:regretLFU}
The LFU algorithm has an expected regret of $O(1)$. More precisely,
\begin{align*}
\mathbb{E}[R(T)] <  \min \left( \frac{16}{\Delta^2_{\min}},   \frac{4 C (L-C)}{\Delta_{\min}} \right)
\end{align*} 
where $\Delta_{\min} =  \mu_{C} - \mu_{C+1}$.
\end{thm}

\begin{remark}
We note that  both terms of the regret upper bound are distribution dependent, i.e., they depend on $\Delta_{\min}$s. Roughly, if $LC < 1/\Delta_{\min}$, then the second term dominates. 
\end{remark}

%
%

We will use the following Lemma for proving Theorem \ref{thm1:regretLFU}.   
\begin{lem}
\begin{align*}
\mathbb{P}(\max_{i} |\hat{\mu}_{i}(t) - \mu_{i}| >\epsilon)  \leq 2e^{-t\epsilon^2/2}
\end{align*} 
\label{lem3:dwf}  
\end{lem}
\begin{proof}
The proof follows from the Dvoretzky-Kiefer-Wolfowitz inequality \cite{dkw}.

Let $F_{i} = \sum^{i}_{j=1} \mu_{j}$ and $\hat{F}_{i}(t) = \sum^{i}_{j=1} \hat{\mu}_{j}(t)$ for $i \in \mathcal{L}$. Define $F_{i} = \hat{F}_{i}(t) = 0$  for $i < 0$ and  $F_{i} = \hat{F}_{i}(t) = 1$ for $i > L$.   {Then, from the Dvoretzky-Kiefer-Wolfowitz inequality \cite{dkw},} for any $\epsilon>0$, we have
\begin{align*}
\mathbb{P}(\max_{i} |\hat{F}_{i}(t) - F_{i} | >\epsilon) \leq 2e^{-2t\epsilon^2}
\end{align*} 
\begin{align*}
\text{Now,}~ &\{\max_{i} |\hat{F}_{i}(t) -F_{i}|>\epsilon/2\}  \\
&= \{\max_{i} |\hat{F}_{i}(t) -F_{i}| + \max_{i} |\hat{F}_{i-1}(t) -F_{i-1}| > \epsilon\}  \\
& \supset \{ \max_{i} | (\hat{F}_{i}(t) -F_{i}) - (\hat{F}_{i-1}(t) -F_{i-1})|   >  \epsilon\} \\
& =  \{\max_{i}  |\hat{\mu}_i(t) -\mu_i| > \epsilon\}  \\
\text{So,}~&\mathbb{P}(\max_{i}  |\hat{\mu}_i(t) -\mu_i| > \epsilon) \leq \mathbb{P}(\max_{i} |\hat{F}_{i}(t) -F_{i}|>\epsilon/2) \\
&\leq 2e^{-t\epsilon^2/2} 
\end{align*}
\end{proof}

We now present the proof of Theorem \ref{thm1:regretLFU}.
\begin{proof}
We denote $C_{\text{LFU}}(t)$ just as $C(t)$ for notational convenience. We first show that  if  $\max_{i} |\hat{\mu}_{i}(t) - \mu_{i} |  < \Delta_{\min}/2$, then $C(t) = \mathcal{C}$. Indeed,  if  $\max_{i} |\hat{\mu}_{i}(t) - \mu_{i} |  < \Delta_{\min}/2$, for any $j \in \mathcal{C}$ and for any $k \in \mathcal{L} \setminus \mathcal{C}$,
\begin{align*}
&\hat{\mu}_{j}(t) \geq  \mu_{j} - \Delta_{\min} / 2 \geq \mu_{C} - \Delta_{\min}/2 \\
&\hspace{1cm} \geq \mu_{C+1} + \Delta_{\min}/2 \geq  \mu_{k} + \Delta_{\min}/2 \geq \hat{\mu}_{k}(t)
\end{align*} 
and hence $C(t) = \mathcal{C}$. Now, 
\begin{align}
&\mathbb{E}[R(T)] = \mathbb{E}[ \sum^{T}_{t=1} \mathbbm{1}\{x(t) \in \mathcal{C}\} - \mathbbm{1}\{x(t) \in C(t)\}]  \nonumber \\
&\leq \mathbb{E}[  \sum^{T}_{t=1} \mathbbm{1}\{C(t) \neq \mathcal{C}\} ] =\sum^{T}_{t=1} \mathbb{P}(C(t) \neq \mathcal{C}) \nonumber \\
&\leq  \sum^{T}_{t=1} \mathbb{P}(\max_{i} |\hat{\mu}_{i}(t) - \mu_{i} |  \geq  \Delta_{\min}/2) \nonumber \\
\label{eq:lfu-r1}
&\leq  \sum^{T}_{t=1} 2 e^{- t \Delta_{\min}^2/8} \leq \frac{16}{\Delta_{\min}^{2}}.  
\end{align} 
We next upper bound $\mathbb{E}[R(T)]$ in a different manner.
\begin{align}
&\mathbb{E}[R(T)] = \mathbb{E}[ \sum^{T}_{t=1} \mathbbm{1}\{x(t) \in \mathcal{C}\} - \mathbbm{1}\{x(t) \in C(t)\}]  \nonumber \\ 
&= \mathbb{E}[ \sum^{T}_{t=1} \mathbb{E}[ \mathbbm{1}\{x(t) \in \mathcal{C}\} | \mathcal{C}] - \mathbb{E}[ \mathbbm{1}\{x(t) \in C(t)\} | C(t)]  \nonumber \\ 
\label{eq:lfu-rstep1}
&=  \mathbb{E} [ \sum_{t=1}^T ( \sum_{j \in \mathcal{C}} \mu_j - \sum_{k \in C(t)} \mu_k ) ]  \\  
\label{eq:lfu-rstep5}
&\leq  \mathbb{E}[ \sum^{T}_{t=1} \sum^{C}_{j=1} \sum^{L}_{k=C+1}\Delta_{j, k} \mathbbm{1}\{j \notin C(t), k \in C(t)\}]   \\ 
&\leq \mathbb{E}[  \sum^{T}_{t=1} \sum^{C}_{j=1} \sum^{L}_{k=C+1}   \Delta_{j, k} \mathbbm{1}(\hat{\mu}_{k}(t) > \hat{\mu}_{j}(t)) ] \nonumber \\ 
&\leq \mathbb{E}[  \sum^{T}_{t=1}  \sum^{C}_{j=1} \sum^{L}_{k=C+1} \Delta_{j, k} \left(\mathbbm{1}\{\hat{\mu}_{j}(t) - \mu_{j} \leq - \Delta_{j,k}/2\} \right.  \nonumber \\
\label{eq:lfu-r2} 
&\hspace{2cm} +\left. \mathbbm{1}\{\hat{\mu}_{k}(t) - \mu_{k} > \Delta_{j,k}/2\} \right) ] 
\end{align}
{Using the Hoeffding inequality \cite{hoeffding1994probability},} we obtain
\begin{align*}
&\mathbb{P}(\hat{\mu}_{j}(t) - \mu_{j} \leq - \Delta_{j,k}/2 ) \leq e^{-t\Delta^{2}_{j,k}/2}, \\
&\mathbb{P}(\hat{\mu}_{k}(t) - \mu_{k} > \Delta_{j,k}/2) \leq e^{-t\Delta^{2}_{j,k}/2}
\end{align*}
Now, continuing form \eqref{eq:lfu-r2} and by taking expectation inside the summation, we obtain
\begin{align}
\mathbb{E}[R(T)] &\leq \sum^{C}_{j=1} \sum^{L}_{k=C+1}  \sum^{T}_{t=1}  \Delta_{j, k} 2 e^{-t\Delta^{2}_{j,k}/2} \nonumber \\ 
\label{eq:lfu-r2}
&\leq  \sum^{C}_{j=1} \sum^{L}_{k=C+1}  \frac{4}{\Delta_{j,k}} \leq \frac{4 C (L-C)}{\Delta_{\min}}
\end{align} 
Combining \eqref{eq:lfu-r1} and \eqref{eq:lfu-r2}, we obtain the desired result. 
\end{proof}

\subsection{WLFU Algorithm}

We showed that the LFU achieves a regret of $O(1)$.  However, the implementation of the LFU is expensive in terms of the memory required, because it maintains the popularity estimate  for each item in the library ($\hat{\mu}_{i}(t)$), and the library size $L$ is extremely large for most practical applications. Hence,  allocating memory to maintain the popularity distribution estimate for the whole library is impractical.

There are many approaches proposed to address this issue \cite{karakostas2002exploitation,einziger2017tinylfu}. However, most approaches rely on heuristics-based approximations of the empirical estimate, often with a tight pre-determined constraint on the memory. This leads to non-optimal use of the available information, and could result in poor performance of the corresponding algorithms.

Here, we consider the Window-LFU (WLFU) algorithm \cite{karakostas2002exploitation} that has been proposed as way to overcome the expensive memory requirement of the LFU.  WLFU employs a sliding window approach. At each time $t$, the algorithm keeps track of only the past $w$ file requests.  This is equivalent to maintaining a time window from $t-w$ to $t$, denoted by $W[t-w, t]$.  Caching decisions are made based on the file requests that appeared in this window.  In  particular, the items to be placed in the cache at time $t$, $C_{WLFU}(t)$, are the top $C$ files with maximum appearances in the window $W[t-w, t]$. 

We now show that the expected cumulative regret incurred by WLFU increases linearly in time ($\Omega(T)$), as opposed to the constant regret ($O(1)$) of the standard LFU. Since  $\Omega(T)$ is the worst possible regret for any learning algorithm, it suggests that in practice there will occasionally be arbitrarily bad sample paths with many misses.

\begin{thm}
Under the WLFU algorithm, $\mathbb{E}[R(T)] = \Omega (T)$.
\label{prop1:WLFU}
\end{thm}

\begin{proof}
%
From the proof of Theorem \ref{thm1:regretLFU} (c.f. \eqref{eq:lfu-rstep1}), we have,
\begin{align} 
&\mathbb{E}[R(T)] =  \mathbb{E} [ \sum_{t=1}^T ( \sum_{j \in \mathcal{C}} \mu_j - \sum_{k \in C(t)} \mu_k ) ] \nonumber  \\
&= \mathbb{E}[\sum_{t=1}^{T}(\sum_{j \in \mathcal{C}\setminus C(t) } \mu_j - \sum_{k \in {C(t)}\setminus \mathcal{C}} \mu_k)] \nonumber \\
&\geq  \mathbb{E} [ \sum_{t=1}^{T}\sum_{k \in  C(t) \setminus \mathcal{C} }(\mu_C-\mu_k ) ] \nonumber \\ 
&= \mathbb{E}[ \sum_{t=1}^{T}\sum_{k \in  \mathcal{L} \setminus \mathcal{C} }(\mu_C-\mu_{k} ) \mathbbm{1} \{k \in C(t)\} ] \nonumber \\ 
&=\sum_{k \in \mathcal{L} \setminus \mathcal{C} }   (\mu_C-\mu_k ) \mathbb{E}[ \sum_{t=1}^{T}  \mathbbm{1} \{k \in C(t)\}] \nonumber \\
&=\sum_{k \in \mathcal{L} \setminus \mathcal{C} }   (\mu_C-\mu_{k} )  \sum_{t=1}^{T}   \mathbb{P}(k \in C(t)) \nonumber \\
\label{eq:wlfu-tep1}
&\geq (\mu_C-\mu_{C+1} )  \sum_{t=1}^{T}   \mathbb{P}(C+ 1 \in C(t))
\end{align} 

\begin{align*}
&\text{We have,}~\mathbb{P}(C+ 1 \in C(t)) \\
&\geq  \mathbb{P}(\text{all requests in}~W[t, t-w]~\text{are}~C+1 )= (\mu_{C+1})^{w}
\end{align*}

Now, continuing from \eqref{eq:wlfu-tep1},  
\begin{align*}
\mathbb{E}[R(T)] \geq (\mu_{C} - \mu_{C+1}) (\mu_{C+1})^{w} T
\end{align*}
\end{proof}

\subsection{LFU-Lite Algorithm}

We now propose a new scheme that we call the LFU-Lite algorithm.  Unlike the LFU algorithm, LFU-Lite algorithm does not maintain an estimate of the popularity  for each item in the library.  Instead, it  maintains the popularity estimate only for a subset of the items that it has observed.  This approach significantly reduces the memory required as compared to the standard LFU implementation.  At the same time, we show that the LFU-Lite achieves an $O(1)$ regret similar to that of the LFU, and thus has a superior performance compared to WLFU which suffers an $\Omega(T)$ regret.   

We achieve this `best of both' performance  by a clever combination of a window based approach to decide the items to maintain an estimate, and by maintaining a separate \textit{counter bank} to keep track of these estimates.  At each time $t$,  LFU-Lite selects the top $C$ items with maximum appearances in the window of observation $W[t-w, t]$.  We denote this set of files as $A(t)$.   Let $B(t-1)$ be the set of items in the counter bank at the beginning of $t$.  Then, if any item $j \in A(t)$ is not present in $B(t-1)$, it is added to the counter bank, and the counter bank is updated to $B(t)$.  Once an item is placed in the counter bank,  it is never removed from the counter bank. 

LFU-Lite maintains an estimate of the popularity of each item in the counter bank. The popularity estimate of item $i \in B(t)$, $\hat{\mu}_{i}(t)$, is defined as 
\begin{align}
\label{eq:muhat-ll}
\hat{\mu}_{i}(t) = \frac{1}{(t-t_{i})}  \sum^{t}_{\tau=t_{i}+1} \mathbbm{1}\{x(t) \in B(t) \}
\end{align}
where $t_{i}$ is the time at which the item $i$ has been added to the counter bank. The item to be placed in the cache at time $t$, $C_{\text{LL}}(t)$, is then selected as
\begin{align*}
C_{\text{LL}}(t) =  \arg \max_{C} (\hat{\mu}_{j}(t), j \in B(t))
\end{align*}
Description of the LFU-Lite is also given in Algorithm~\ref{alg:lfulite}.

\begin{algorithm}
\caption{LFU-Lite} 
\label{alg:lfulite}
\begin{algorithmic}
\FOR{$t=1, \dots, T$}
\STATE Observe $x(t)$ 
\STATE Select $A(t)$, the top $C$ files with maximum appearances  in the window $W[(t-w)_{+}, t]$
\FOR {Each $j \in A(t)$ } 
\IF {($j \in A(t)$ is not in $B(t-1)$)}
\STATE  $t_{j} \leftarrow t$
\STATE Add file $j$ into $B(t)$
\ENDIF
\ENDFOR
\STATE Select the files $C_{\text{LL}}(t)  = \arg \max_{C} (\hat{\mu}_{j}(t), j \in B(t))$ and place them in the cache
\ENDFOR
\end{algorithmic}
\end{algorithm}

We now present the performance guarantee for the LFU-Lite algorithm.

\begin{thm}
The expected regret under the LFU-Lite algorithm is
\begin{align*} 
\mathbb{E}[R(T)] \leq \frac{C (L-C) w}{p_{\min}}  + \frac{4 C (L-C) }{\Delta_{\min}},
\end{align*}
where $\Delta_{\min} = \mu_{C} - \mu_{C+1}$, $
p_{\min} = \sum_{n=\mu_{C+1}w+1}^{w} {w \choose n} \mu_{C}^n(1-\mu_{C})^{w-n}$. 
\end{thm}

\begin{proof}
%
For each item $i \in \mathcal{L}$, $\hat{\mu}_{i}(t)$ is defined as in \eqref{eq:muhat-ll} for $t > t_{i}$.  Here, we also define $\hat{\mu}_{i}(t) = 0$ for $t \leq t_{i}$, before item $i$ enters the counter bank.  We note that this is only a proof approach and doesn't influence the implementation of the algorithm. Now, from \eqref{eq:lfu-rstep5} 
\begin{align}
\mathbb{E}[R(T)] &\leq   \mathbb{E}[ \sum^{T}_{t=1} \sum^{C}_{j=1} \sum^{L}_{k=C+1}\Delta_{j, k} \mathbbm{1}\{j \notin C(t), k \in C(t)\}]  \nonumber \\ 
\label{eq:reg-ll-st1}
&\leq \mathbb{E}[  \sum^{T}_{t=1} \sum^{C}_{j=1} \sum^{L}_{k=C+1}   \Delta_{j, k} \mathbbm{1}(\hat{\mu}_{k}(t) > \hat{\mu}_{j}(t)) ]  \\
&\leq \mathbb{E}[  \sum^{T}_{t=1}  \sum^{C}_{j=1} \sum^{L}_{k=C+1} \Delta_{j, k} \left(\mathbbm{1}\{\hat{\mu}_{j}(t) - \mu_{j} \leq - \Delta_{j,k}/2\} \right.  \nonumber \\ 
\label{eq:reg-ll-st2}
&\hspace{2cm} +\left. \mathbbm{1}\{\hat{\mu}_{k}(t) - \mu_{k} > \Delta_{j,k}/2\} \right) ].
\end{align}   
Note that the LFU-Lite algorithm incurs a regret at time $t$ if an item $j \in \mathcal{C}$ is not present in the counter bank $B(t)$.  This is  taken into account in the above expression (c.f. \eqref{eq:reg-ll-st1})  by defining $\hat{\mu}_{j}(t) = 0$ for $j \notin B(t)$. 

The first term in  \eqref{eq:reg-ll-st2} can be bounded as
\begin{align}
&\mathbb{E}[\sum^{T}_{t=1}  \sum^{C}_{j=1} \sum^{L}_{k=C+1} \Delta_{j, k} \mathbbm{1}\{\hat{\mu}_{j}(t) - \mu_{j} \leq - \Delta_{j,k}/2\}   ] \nonumber \\
&= \mathbb{E}[ \sum^{C}_{j=1} \sum^{L}_{k=C+1} \mathbb{E}[\sum^{T}_{t=1} \Delta_{j, k} \mathbbm{1}\{\hat{\mu}_{j}(t) - \mu_{j} \leq - \Delta_{j,k}/2\} | t_{j}]] \nonumber \\
&\leq \mathbb{E}[ \sum^{C}_{j=1} \sum^{L}_{k=C+1} \Delta_{j,k}(  t_{j}  + \mathbb{E}[\sum^{T}_{t=t_{j}} \mathbbm{1}\{\hat{\mu}_{j}(t) - \mu_{j} \leq - \Delta_{j,k}/2\} | t_{j}] )] \nonumber \\
&\leq \mathbb{E}[ \sum^{C}_{j=1} \sum^{L}_{k=C+1} \Delta_{j,k}  ( t_{j}  + \sum^{T}_{t=t_{j}} \mathbb{P}(\hat{\mu}_{j}(t) - \mu_{j} \leq - \Delta_{j,k}/2 | t_{j}) )]  \nonumber \\ 
&\leq \mathbb{E}[ \sum^{C}_{j=1} \sum^{L}_{k=C+1} \Delta_{j,k} (  t_{j}  + \sum^{T}_{t=t_{j}} e^{-(t-t_{j}) \Delta^{2}_{j,k}/2} ]  \nonumber \\
\label{eq:reg-ll-st3}
&\leq  \sum^{C}_{j=1} \sum^{L}_{k=C+1}  ( \Delta_{j,k} \mathbb{E}[ t_{j} ]  +  \frac{2}{\Delta_{j,k}}).
\end{align}  
Similarly, the second term in  \eqref{eq:reg-ll-st2} can be bounded as
\begin{align}
&\mathbb{E}[\sum^{T}_{t=1}  \sum^{C}_{j=1} \sum^{L}_{k=C+1} \Delta_{j, k} \mathbbm{1}\{\hat{\mu}_{k}(t) - \mu_{k} > \Delta_{j,k}/2\}   ] \nonumber \\
&= \mathbb{E}[ \sum^{C}_{j=1} \sum^{L}_{k=C+1} \mathbb{E}[\sum^{T}_{t=1} \Delta_{j, k} \mathbbm{1}\{\hat{\mu}_{k}(t) - \mu_{k} > \Delta_{j,k}/2\} | t_{k}]] \nonumber \\
&= \mathbb{E}[ \sum^{C}_{j=1} \sum^{L}_{k=C+1} \Delta_{j,k}  \sum^{T}_{t=t_{k}} \mathbb{P}(\hat{\mu}_{k}(t) - \mu_{k} > \Delta_{j,k}/2 | t_{k} ) ] \nonumber \\
&\leq  \mathbb{E}[ \sum^{C}_{j=1} \sum^{L}_{k=C+1} \Delta_{j,k}  \sum^{T}_{t=t_{k}} e^{-(t-t_{k}) \Delta^{2}_{j,k}/2} ]  \nonumber \\ 
\label{eq:reg-ll-st4}
&\leq  \sum^{C}_{j=1} \sum^{L}_{k=C+1}   \frac{2}{\Delta_{j,k}} .
\end{align}    
Combining \eqref{eq:reg-ll-st3} and  \eqref{eq:reg-ll-st4} we obtain
\begin{align}
\label{eq:reg-ll-st6}
\mathbb{E}[R(T)] &\leq  \sum^{C}_{j=1} \sum^{L}_{k=C+1}  \Delta_{j,k} \mathbb{E}[ t_{j} ]    +  \sum^{C}_{j=1} \sum^{L}_{k=C+1}   \frac{4}{\Delta_{j,k}}
\end{align}

%

%
It only then remains to bound  $\mathbb{E}[ t_{j} ]$ for $j \in \mathcal{C},$ which can easily be shown to satisfy   
\begin{align} \label{eq:reg-ll-st5}
\mathbb{E}[ t_{j} ] \leq \sum^{\infty}_{t=1}  (1-p_{j})^{\lceil t/w \rceil} \leq  \sum^{\infty}_{k=1} w  (1-p_{j})^{k } \leq w /p_{j}. 
\end{align} 
where $p_{j}$ is the probability that item $j$ is selected in a given window.  

Combining \eqref{eq:reg-ll-st6}  and  \eqref{eq:reg-ll-st5}, we obtain
\begin{align}
\label{eq:reg-ll-st6}
\mathbb{E}[R(T)] &\leq  \sum^{C}_{j=1} \sum^{L}_{k=C+1}  \frac{w \Delta_{j,k}}{p_{j}}    +  \sum^{C}_{j=1} \sum^{L}_{k=C+1}   \frac{4}{\Delta_{j,k}} \nonumber \\
&\leq \frac{C (L-C) w}{p_{\min}}  + \frac{4 C (L-C) }{\Delta_{\min}}.
\end{align} 
\end{proof}  

\begin{prop}
The growth of the expected size of the counter bank as a function of time is concave.
\end{prop} 
\begin{proof}
Let $\bar{p}_j^t$ be the probability with which file $i$ enters the counter bank by time $t$. Note that $\bar{p}_i^t= 1-(1-p_i)^{{\lceil t/w \rceil}}$ . Where $p_i = \sum_{n=\mu_{C+1}w+1}^{w} {w \choose n} \mu_{i}^n(1-\mu_{i})^{w-n}$ is the probability  that item $i$ enters the counter bank in any given window. 
\begin{align}
\label{eq:CB1}
\mathbb{E}[B(t)]&=\mathbb{E}[\sum_{i=1}^{L}\mathbbm{1}_{\{i \in B(t)\}}]  =\sum_{i=1}^{L} \bar{p}_i^t = \sum_{i=1}^{L} 1-(1-p_i)^{{\lceil t/w \rceil}}
\end{align}
Observe that $ 1-(1-p_i)^{{\lceil t/w \rceil}}$ is concave in $t$ and ${E}[B(t)]$ is a sum of $L$ concave functions, and is hence concave.  
\end{proof}

\begin{remark}
Intuitively, the counter bank will keep the counts only for more popular items. It is also straight forward to show that the expected size of the counter bank decreases with the window length $w$. To see this, consider two different window length $w_{1}$ and $w_{2}$ such that $w_{1} \geq w_{2}$. Let $p_{i}(w)$ be the probability that item $i$ enters the counter bank in any given window when the window of length is $w$. Then, for $i \notin \mathcal{C}$, $p_i(w_1) \leq p_i(w_2)$. Intuitively, larger window length leads to more observations and hence to smaller  probability of observing item $i \notin \mathcal{C}$  more than the threshold $\mu_{C+1}w$.  Now, $(1-p_{i}(w_2))^{{\lceil t/w_2 \rceil}} \leq (1-p_{i}(w_1))^{{\lceil t/w_2 \rceil}} \leq (1-p_{i}(w_1))^{{\lceil t/w_1 \rceil}}$. So, $1 - (1-p_{i}(w_2))^{{\lceil t/w_2 \rceil}} \geq 1- (1-p_{i}(w_1))^{{\lceil t/w_1 \rceil}}$. Hence, the contribution of $i \notin \mathcal{C}$ to the expected size of the counter bank according \eqref{eq:CB1} is smaller for larger window length. The exact dependence of $\mathbb{E}[B(t)]$ on $w$ is cumbersome to characterize. We, however, illustrate this  through extensive simulations in Section \ref{section:Simulations}.
\end{remark}

%
%


\section{Caching with Partial Observation}

We now consider the problem of optimal caching under the partial observation regime. As described earlier, here the algorithm can observe a file request only if the requested file is already in the cache.  Hence, the caching algorithm has to perform an active \textit{exploration} by placing a file in the cache for a sufficiently long time to learn its popularity, in order to decide if that file belongs to the set of the most popular files. This procedure is in sharp contrast to the full observation structure where the popularity of each file in the library can be learned (improved) after each time step.  

However, the exploration is costly because the algorithm incurs regret every time that a sub-optimal file is placed in the cache for exploration.  Hence, the algorithm also has to perform an active \textit{exploitation}, i.e., place the most popular items according to the current estimate in the cache.  The optimal \textit{exploration vs exploitation} trade-off for minimizing the regret is at the core of most online learning algorithms.  The Multi-Armed Bandit (MAB) model is a canonical formalism for this class of problems.  Here, there are multiple arms (actions) that yield random rewards independently over time, with the (unknown) mean of arm $i$ being $\mu_{i}$. The objective is to learn the mean reward of each arm by exploration and maximize cumulative reward  by exploitation.

\subsection{Caching Bandit with Full Posterior Sampling} 

Posterior sampling based algorithms for MAB \cite{agrawal2013further, agrawal2013thompson} typically use a Beta prior (with Bernoulli likelihood) or Gaussian prior (with Gaussian likelihood) in order to exploit the the conjugate pair property of the prior and likelihood (reward) distributions. Hence, the posterior at any time will  have the same form as the prior distribution, albeit with different parameters. This provides a computationally tractable and memory efficient way to keep track of the posterior distribution evolution. However, in the optimal caching problem, the unknown popularity vector $\mu$ has interdependent components through the constraint $\sum_{i} \mu_{i}=1$. Hence, standard prior distributions like Beta will not be able to capture the full posterior evolution in the caching problem.   

We use a Dirichilet prior on the popularity distribution $\mu = (\mu_{1}, \ldots, \mu_{L})$, parametrized by $\alpha = (\alpha_{1}, \ldots, \alpha_{L})$. More precisely,
\begin{equation*}
f_{0}(\mu ; \alpha) = \frac{1}{B(\alpha)}\prod_{i=1}^L \mu_i^{\alpha_i - 1},~\text{where,}~ B(\alpha) = \frac{\prod_{i=1}^L\Gamma(\alpha_i)} {\Gamma(\sum_{i=1}^L \alpha_i)} 
\end{equation*}   
and $\Gamma (\cdot)$ is the Gamma function.   
 
Let $f_{t}$ be the posterior distribution at time $t$ with parameter $\alpha(t)$. The posterior is updated according to the observed information $s(t)$.  In the case of a hit, 
 the file request $x(t)$ is observed and $s(t) = x(t)$. It is easy to see that the correct posterior update is $\alpha(t) = \alpha(t) + e_{x(t)}$, where  $e_{x(t)}$ is the unit vector with non-zero element at index $x(t)$. 

The posterior update is complex in the case of the cache miss. Given the current parameter $\alpha(t) = \alpha$, the probability of a cache miss is given by 
\begin{align*} 
\mathbb{P}(s(t) = 0) &= \sum_{j \notin C(t)} \Pr(x(t)=j) \\
&=\sum_{j \notin C(t)} \int_S \mu_j \frac{\prod_{i=1}^K \mu_i^{\alpha_i-1}}{B(\alpha)} dP =\frac{\sum_{j \notin C(t)}\alpha_j}{\sum_{i=1}^L \alpha_i}
\end{align*}
Then we can show that the posterior distribution can be computed as

\begin{align*}
&f(\mu| s(t) = 0) = \mathbb{P}(s(t) = 0|\mu) f(\mu;\alpha) / \mathbb{P}(s(t) = 0) \\
&= \frac{1}{\sum_{j \notin C(t)} \alpha_j } \sum_{j \notin C(t)} \alpha_j f(\mu ; \alpha+e_j)
\end{align*}  

\begin{algorithm}
\caption{CB-FPS Algorithm} 
\label{alg:fps} 
\begin{algorithmic}
\STATE \text{Initialize the prior distribution $f_{0}$ }
\FOR{$t = 1, \dots, T$} 
\STATE Sample $\hat{\mu}(t) \sim f_{t}(\cdot) $
\STATE Select $C_{\text{FPS}}(t) = \arg \max_{C} \hat{\mu}(t)$
\STATE Receive the observation $s(t)$
\STATE Update the posterior $f_{t+1}(\mu) \propto \mathbb{P}(s(t)|\mu) f_{t}(\mu)$
\ENDFOR
\end{algorithmic}
\end{algorithm}

Hence, with each miss, the algorithm needs to store a set of size $(L-C)$ consisting of Dirichlet parameters.  As the number of misses increases, the memory required to store these parameters will increase exponentially. Hence, the full posterior update algorithm is infeasible from an implementation perspective.   

We  present the CB-FPS in Algorithm \ref{alg:fps}.  
In Section~\ref{section:Simulations}, we will see that a Monte Carlo version of this algorithm  can be implemented for small values of $L$ and $C$, which seems to achieve an $O(1)$ regret.  A rigorous proof that shows such regret, even in some special cases and by neglecting computational tractability, is an interesting open problem.

\subsection{Caching Bandit with Marginal Posterior Sampling} 

We now propose an algorithm that only performs a marginal posterior update.  Instead of maintaining a Dirchlet prior for the popularity vector $\mu$, we use a Beta prior for the popularity of each \emph{individual} item $\mu_{i}$.  
The CB-MPS is described in Algorithm~\ref{alg:mps}.     

\begin{algorithm}
\caption{CB-MPS Algorithm} 
\label{alg:mps} 
\begin{algorithmic}
\STATE \text{Initialize $\alpha_i(0) = 1, \beta_i(0) = 1, \forall i \in \mathcal{L}.$ }
\FOR{$t = 1, \dots, T$}
\STATE Generate samples $\hat{\mu}_i(t) \sim \text{Beta}(\alpha_i(t),\beta_i(t))$
\STATE $C_{\text{MPS}}(t) \gets \arg \max_{C} \hat{\mu}(t)$ 
\IF {$x(t) \in C(t)$}
\STATE $\alpha_{x(t)}(t+1) \leftarrow \alpha_{x(t)}(t) + 1$
\STATE $\beta_{i}(t+1) \leftarrow \beta_{i}(t) + 1, \forall i \in C(t), i \neq x(t)$
\ENDIF 
\ENDFOR
\end{algorithmic}
\end{algorithm}

We now provide a performance guarantee for the CB-MPS algorithm.
\begin{thm}
Under marginal posterior sampling algorithm, $\mathbb{E}[R(T)] = O((L-C) C \log T)$.
\end{thm}

We omit the proof of this theorem because the analysis is  similar to that of multi-payer multi-armed bandit algorithm. In particular, the posterior sampling method proposed in \cite{komiyama2015optimal} can be used with small modifications to show the above result.

%

\subsection{Caching Bandit with Structural Information}

Even though the MPU algorithm is easy to implement, it suffers an $O(\log T)$ regret, which is much worse than the $O(1)$ regret incurred by LFU and LFU-Lite. This is due to  the partial observation structure that limits the rate of learning.  We now propose an algorithm that we call Caching Bandit with Structural Information (CB-SI). We show that with a minimal assumption on the availability of the structural information about the popularity distribution, CB-SI can achieve an $O(1)$ regret even in the partial observation regime.  

We assume that the algorithm knows the value of $\mu_{C}$ and $\Delta_{\min}$, the popularity value of the $C$th most popular item and the optimality gap.  Note that we do not assume knowledge of the identity of the $C$th most popular file. We note that our proof approach follows the techniques developed in \cite{bb13}, which can be considered as a special case with $C=1$. 
CB-SI algorithm is given in Algorithm \ref{alg:cb}.   
\begin{algorithm}
\caption{CB-SI Algorithm} 
\label{alg:cb}
\begin{algorithmic}
\STATE {Initialize $\alpha_i(0) = 0, \beta_i(0) = 0, \hat{\mu}_{i}(0) = 1/L, \forall i \in \mathcal{L}.$ }
\STATE Initialize $n_{i}(t) = 0, \forall i \in \mathcal{L}$
\FOR{$t = 1, \dots, T$}
\STATE Compute the set $A(t) = \{  i \in \mathcal{L}: \hat{\mu}_{i}(n_{i}(t)) \ge \mu_C -\frac{\Delta}{2} \}$
\IF {$|A(t)| \geq C$}
\STATE $C_{SI}(t)  = \argmax_{C} \hat{\mu}_{j}(n_{j}(t), j \in A(t))$
\STATE $Z_t \gets 1$
\ELSE
\STATE For each $i \in \mathcal{L} \setminus A(t)$, compute 
\begin{align*}
&p_{i}(t) = c/(\mu_C - \hat{\mu}_{i}(n_{i}(t)))^2 \\
&c = \sum_{i \in \mathcal{L} \setminus C(t)} 1/(\mu_C - \hat{\mu}_{i}(n_{i}(t)))^{2}
\end{align*} 
\STATE Sample $C - |A(t)|$ elements from the set $\mathcal{L} \setminus A(t)$ according to the probability $p_{i}(t)$. Denote these elements as $B(t)$
\STATE $C_{\text{SI}}(t) = A(t) \cup B(t)$
\STATE $Z_t \gets 2$
\ENDIF
\STATE Place the files $C_{\text{SI}}(t)$ in the cache
\IF {$x(t) \in C(t)$}
\STATE $\alpha_{x(t)}(t+1) \leftarrow \alpha_{x(t)}(t) + 1$
\STATE $\beta_{i}(t+1) \leftarrow \beta_{i}(t) + 1, \forall i \in C(t), i \neq x(t)$
\ENDIF 
\STATE $n_{i}(t+1) = \alpha_{i}(t+1)+\beta_{i}(t+1)$
\STATE  $\hat{\mu}_{i}(n_{i}(t+1)) =  \alpha_{i}(t+1)/n_{i}(t+1) $
\ENDFOR
\end{algorithmic} 
\end{algorithm}


\begin{thm}
The expected cumulative regret of CB-SI Algorithm is, 
\begin{align*}
\mathbb{E}[R(T)] \leq C \sum_{j \in \mathcal{L} \setminus \mathcal{C}} ( \frac{2}{\Delta^2} + \frac{4} {\Delta^2} [4+\frac{32}{\Delta^2} \exp(-\frac{\Delta^2}{8})] ),
\end{align*}
where $\Delta = \mu_C - \mu_{C+1}$.
\end{thm}
\begin{proof}
We denote $\Delta_j = \mu_C - \mu_j$. In this proof, we will show that the expected regret of CB-SI algorithm is bounded above by a constant.
From the proof of Theorem \ref{thm1:regretLFU} (c.f. \eqref{eq:lfu-r2}), we get that the expected regret is bounded as below.

\begin{align}\label{L}
&\mathbb{E}[R(T)] \leq \mathbb{E}[ \sum^{T}_{t=1} \sum^{C}_{i=1} \sum^{L}_{j=C+1}\Delta_{j, k} \mathbbm{1}\{i \notin C(t), j \in C(t)\}]  \nonumber \\ 
&\leq  C ~ \mathbb{E}[  \sum^{T}_{t=1} \sum^{L}_{j=C+1}  \mathbbm{1}\{j \in C(t)\} ] =  C  \sum^{L}_{j=C+1}\sum^{T}_{t=1}   \mathbb{P}(j \in C(t))  \nonumber  \\
&=   C  \sum^{L}_{j=C+1}\sum^{T}_{t=1}  ( \mathbb{P}(\hat{\mu}_{j}(n_{j}(t))  > \mu_{C} - \Delta_{j}/2, j \in C(t) )  \nonumber  \\
&\hspace{1.5cm}  + \mathbb{P}(\hat{\mu}_{j}(n_{j}(t))  \leq \mu_{C} - \Delta_{j}/2, j \in C(t) )).   
\end{align}   
We address each term in the above summation separately. First, observe that for any $j \in \{C+1 \dots L\},$ the following inequality holds.
\begin{align*}
&\sum^{T}_{t=1}  \mathbb{P}(\hat{\mu}_{j}(n_{j}(t))  > \mu_{C} - \Delta_{j}/2, j \in C(t) )\\
&\leq \sum^{T}_{t=1}  \mathbb{P}(\hat{\mu}_{j}(n_{j}(t))  > \mu_{j} + \Delta_{j}/2, j \in C(t) )\\ 
&\leq  \sum^{T}_{t=1}  \mathbb{P}(\hat{\mu}_{j}(t)  > \mu_{j} + \Delta_{j}/2 )  \stackrel{(a)} \leq \sum^{T}_{t=1} e^{-\Delta^{2}_{j}t/2} \leq \frac{2}{\Delta^{2}_{j}}, \numberthis \label{N}
\end{align*}
where the inequality $(a)$ follows from Hoeffding's inequality.\\
For bounding the second term in \eqref{L}, we use the policy definition. Since $\Delta_j \ge \Delta$, the first inequality follows trivially. The equality$(b)$ follows from the fact that when the mean estimate of the $j^{th}$ item is smaller than $\mu_C - \frac{\Delta}{2}$, the only means by which it can enter the cache is through exploration part of the algorithm, which is denoted by $Z_t = 2.$
\begin{align*}
&\mathbb{P} (\hat{\mu}_j(n_j(t)) \leq \mu_C-\frac{\Delta_j}{2}, j \in C(t)) \\
&\leq \mathbb{P} (\hat{\mu}_j(n_j(t)) \leq \mu_C-\frac{\Delta}{2}, j \in C(t)) \\
&\stackrel{(b)}= \mathbb{P} (\hat{\mu}_j(n_j(t)) \leq \mu_C-\frac{\Delta}{2}, j \in C(t),Z_t=2) \\
&\stackrel{(c)}=\mathbb{P} ( j \in C(t) | \hat{\mu}_j(n_j(t)) \leq \mu_C-\frac{\Delta}{2},Z_t=2)  \\
&\hspace{50pt} \mathbb{P} ( \hat{\mu}_j(n_j(t)) \leq \mu_C-\frac{\Delta}{2}, Z_t=2) \\
&= p_{j,t} \mathbb{P} \left( \hat{\mu}_j(n_j(t)) \leq \mu_C-\frac{\Delta}{2}, Z_t=2 \right)  \\
&= \mathbb{E} \left[p_{j,t} \mathbbm{1} \{ \hat{\mu}_j(n_j(t)) \leq \mu_C-\frac{\Delta}{2}, Z_t=2\} \right]  \\
&=  \mathbb{E} \left[ \frac{p_{j,t}}{p_{i,t}} p_{i,t} \mathbbm{1} \{ \hat{\mu}_j(n_j(t)) \leq \mu_C-\frac{\Delta}{2}, Z_t=2\}\right] ,
\end{align*}
for any $i \in \mathcal{C}$.
Note that, in the equality${(c)}$, we used the definition of $p_{j,t}$.
Now, substituting the value for the sampling probability $p_{i,t}$, we obtain,
\begin{align*}
& \leq \mathbb{E} \left[ \frac{|\mu_C-\hat{\mu}_i(n_i(t))|^2} {(\frac{\Delta}{2})^2} p_{i,t} 1 \{ \hat{\mu}_j(n_j(t)) \leq \mu_C-\frac{\Delta}{2}, Z_t=2\}\right]\\
& \leq \frac{4} {\Delta^2} \mathbb{E} \left[ |\mu_C-\hat{\mu}_i(n_i(t))|^2 p_{i,t} \mathbbm{1} \{ Z_t=2\} \right]  \\
& \leq \frac{4} {\Delta^2} \mathbb{E} \left[ |\mu_C-\hat{\mu}_i(n_i(t))|^2 \right. \\
& \left. \hspace{40pt} \mathbb{P} \left(i \in C(t) | \hat{\mu}_i(n_i(t)) \leq \mu_C-\frac{\Delta}{2}, Z_t=2 \right)\right]  \\
& = \frac{4} {\Delta^2} \mathbb{E} \left[ |\mu_C - \hat{\mu}_i( n_i(t) )|^2 \right.\\
& \left. \hspace{40pt} \mathbb{E} \left[ \mathbbm{1}\{i \in C(t)\} | \hat{\mu}_i(n_i(t)) < \mu_C-\frac{\Delta}{2}, Z_t=2 \right] \right]  \\
& = \frac{4} {\Delta^2} \mathbb{E} \left[ |\mu_C-\hat{\mu}_i(n_i(t))|^2  \right.\\
&  \left. \hspace{40pt} \mathbb{E} \left[ \mathbbm{1} \{ i \in C(t), \hat{\mu}_i(n_i(t)) < \mu_C-\frac{\Delta}{2}\} | \right. \right.\\
&  \left. \left.\hspace{60pt} \{\hat{\mu}_i(n_i(t)) < \mu_C-\frac{\Delta}{2}, Z_t=2 \} \right] \right]  \\
& = \frac{4} {\Delta^2} \mathbb{E} \left[ |\mu_C-\hat{\mu}_i(n_i(t))|^2 \right. \\
&\left. \hspace{40pt} \mathbbm{1} \{ i \in C(t), \hat{\mu}_i(n_i(t)) < \mu_C-\frac{\Delta}{2}\} \right] \numberthis \label{M}.
\end{align*}
Here, the inequalities follow from the properties of conditional expectation.
Now, we obtain a bound for the second term in \eqref{L}, using \eqref{M}, as below.
\begin{align*}
&\sum_{t=1}^T \mathbb{E} \left[ |\mu_C-\hat{\mu}_i(n_i(t))|^2 1\{\hat{\mu}_i(n_i(t)) < \mu_C-\frac{\Delta}{2}, i \in C(t)\} \right] \\
&\leq \sum_{t=1}^T \mathbb{E} \left[ |\mu_C-\hat{\mu}_{i}(t)|^2 1\{|\hat{\mu}_{i}(t)-\mu_C| >\frac{\Delta}{2}\} \right]\\
&= \sum_{t=1}^T \int_{0}^{\infty} \mathbb{P} ( |\mu_C-\hat{\mu}_{i}(t)|^2 1\{|\hat{\mu}_{i}(t)-\mu_C| >\frac{\Delta}{2} \} \ge x ) dx \\
&=\sum_{t=1}^T \int_{0}^{\infty} \mathbb{P} ( |\mu_C-\hat{\mu}_{i}(t)|^2 1\{|\hat{\mu}_{i}(t)-\mu_C| ^2 >\frac{\Delta^2}{4} \} \ge x ) dx \\
&= \sum_{t=1}^T \left[ \int_{0}^{\frac{\Delta^2}{4}} \mathbb{P} \left(  |\mu_C-\hat{\mu}_{i}(t)|^2 \right.\right.\\
& \left. \left. \hspace{70pt} 1\{|\hat{\mu}_{i}(t)-\mu_C| ^2 >\frac{\Delta^2}{4} \} \ge x \right) dx \right. \\
& \left. \hspace{10pt}+ \int_{\frac{\Delta^2}{4}}^{\infty}  \mathbb{P} ( |\mu_C-\hat{\mu}_{i}(t)|^2 1\{|\hat{\mu}_{i}(t)-\mu_C| ^2 >\frac{\Delta^2}{4} \} \ge x ) dx \right]\\
&= \sum_{t=1}^T  \left[ \int_{0}^{\frac{\Delta^2}{4}} \mathbb{P} ( |\mu_C-\hat{\mu}_{i}(t)|^2  \ge x, |\hat{\mu}_{i}(t)-\mu_C| ^2 >\frac{\Delta^2}{4} ) dx  \right. \\
&  \left. \hspace{10pt} + \int_{\frac{\Delta^2}{4}}^{\infty}  \mathbb{P} ( |\mu_C-\hat{\mu}_{i}(t)|^2 1\{|\hat{\mu}_{i}(t)-\mu_C| ^2 >\frac{\Delta^2}{4} \} \ge x ) dx \right]\\
&= \sum_{t=1}^T  \left[ \int_{0}^{\frac{\Delta^2}{4}} \mathbb{P}\left(|\hat{\mu}_{i}(t)-\mu_C| ^2 >\frac{\Delta^2}{4} \right) dx  \right. \\
&  \left. \hspace{10pt} + \int_{\frac{\Delta^2}{4}}^{\infty}  \mathbb{P}\left(|\mu_C-\hat{\mu}_{i}(t)|^2 1\{|\hat{\mu}_{i}(t)-\mu_C| ^2 >\frac{\Delta^2}{4} \} \ge x \right) dx \right]\\
&= \sum_{t=1}^T  \left[ \frac{\Delta^2}{4} \mathbb{P} \left(|\hat{\mu}_{i}(t)-\mu_C| ^2 >\frac{\Delta^2}{4} \right) + \right.\\
& \left. \hspace{20pt} \int_{\frac{\Delta^2}{4}}^{\infty}  \mathbb{P}\left(|\mu_C-\hat{\mu}_{i}(t)|^2 1\{|\hat{\mu}_{i}(t)-\mu_C| ^2 >\frac{\Delta^2}{4} \} \ge x \right) dx \right]\\
&= \sum_{t=1}^T  \left[ 2 \frac{\Delta^2}{4} e^{-\frac{t\Delta^2}{8}} + \int_{\frac{\Delta^2}{4}}^{\infty}  \Pr\{ |\mu_C-\hat{\mu}_{i}(t)|^2 \ge x \} dx \right]\\
&\leq 4+\frac{32}{\Delta^2} \exp(-\frac{\Delta^2}{8}) \numberthis \label{P}
\end{align*} 
Combining equations \eqref{L},\eqref{M},\eqref{N},\eqref{P},  we observe:
$$\mathbb{E}[R(T)] \leq C \sum_{j \in \mathcal{L} \setminus \mathcal{C}} \left[ \frac{2}{\Delta^2} + \frac{4} {\Delta^2} [4+\frac{32}{\Delta^2} \exp(-\frac{\Delta^2}{8})] \right]$$
\end{proof}

\begin{rem}
We introduce another version of CB-SI algorithm, which is similar in spirit to LFULite. Following a similar rule to LFULite, we maintain a window of the $W$ past observations, and at each time, the $C$ most frequently requested items in the window are added to the counter bank, if those items are not already present in it.  The mean estimates of CB-SI are calculated only for items in the counter bank.  
We call this algorithm as CB-SILite. In Section~\ref{subsec:traces}, we will observe that CB-SILite drastically reduces the number of counters needed to give a similar hit performance to CB-SI.  
\end{rem}
\section{Simulations}
\label{section:Simulations}
In this section, we start by conducting simulations with requests generated under the IRM model to verify the insights on regret obtained in the earlier sections.  We then use two data traces to compare the performance of our proposed algorithms when exposed to a non-stationary arrival process.   Since these requests change with time, we modify the algorithms to "forget" counts, by halving the counts at a fixed periodicity.  In the full observation regime, we also compare the performance against LRU, which is widely deployed and implicitly has a finite memory (i.e., it automatically "forgets").  We also further explore the reaction of our approaches to non-stationary requests by creating a synthetic trace that exhibits changes at a faster timescale than the data traces. 

\begin{figure*}[htbp]
\centering
\begin{minipage}{.32\textwidth}
\centering
\includegraphics[width=1\columnwidth]{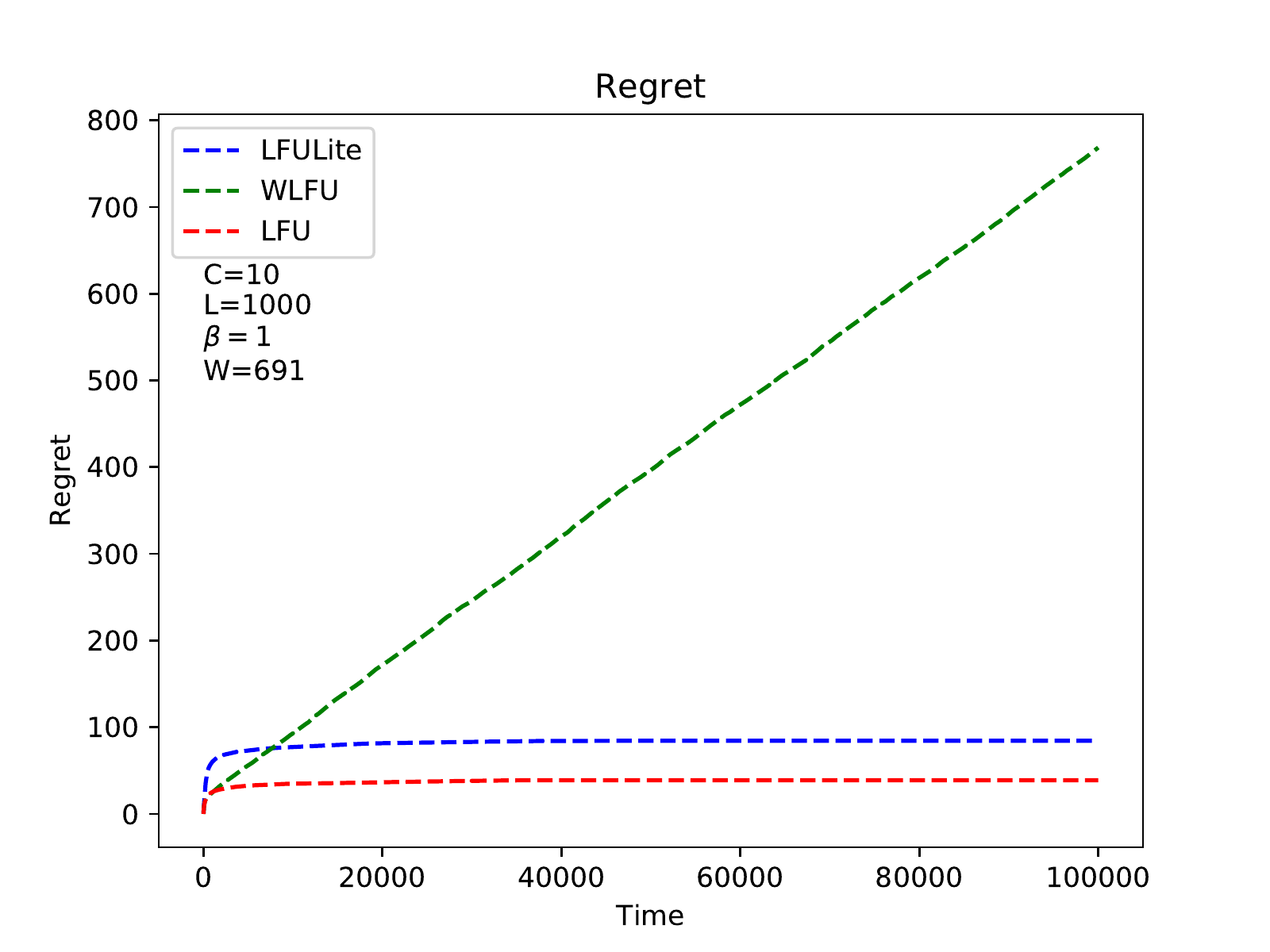}
\caption{Regret of LFU, WLFU, LFU-Lite}
\label{fig:regret_compare}
\end{minipage}\hfill
\begin{minipage}{.32\textwidth}
\centering
\includegraphics[width=1\columnwidth]{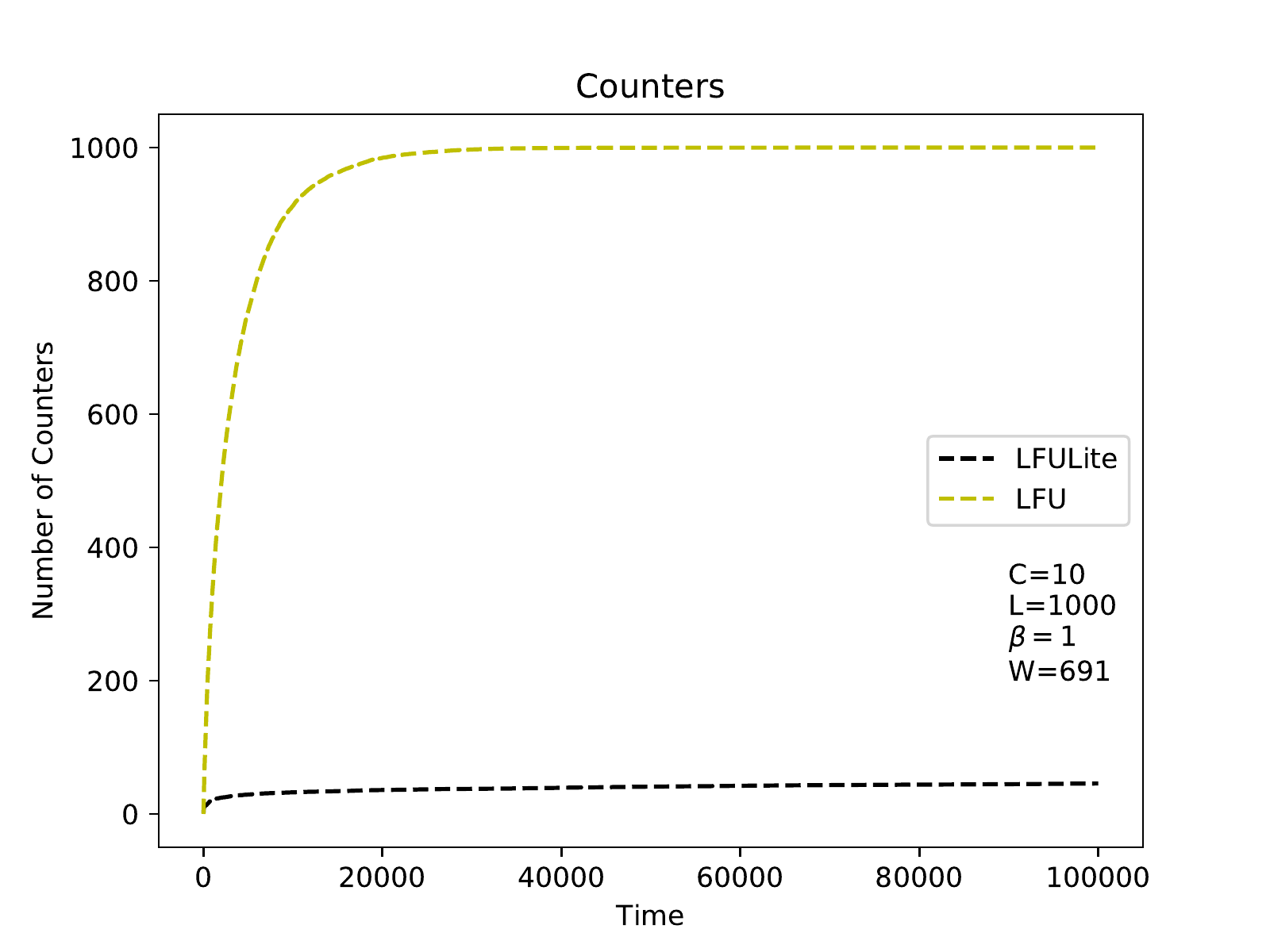}
\caption{Growth of Counters for LFU, LFU-Lite}
\label{fig:counter_compare}
\end{minipage}\hfill
\centering
\begin{minipage}{.32\textwidth}
\centering
\includegraphics[width=1\columnwidth]{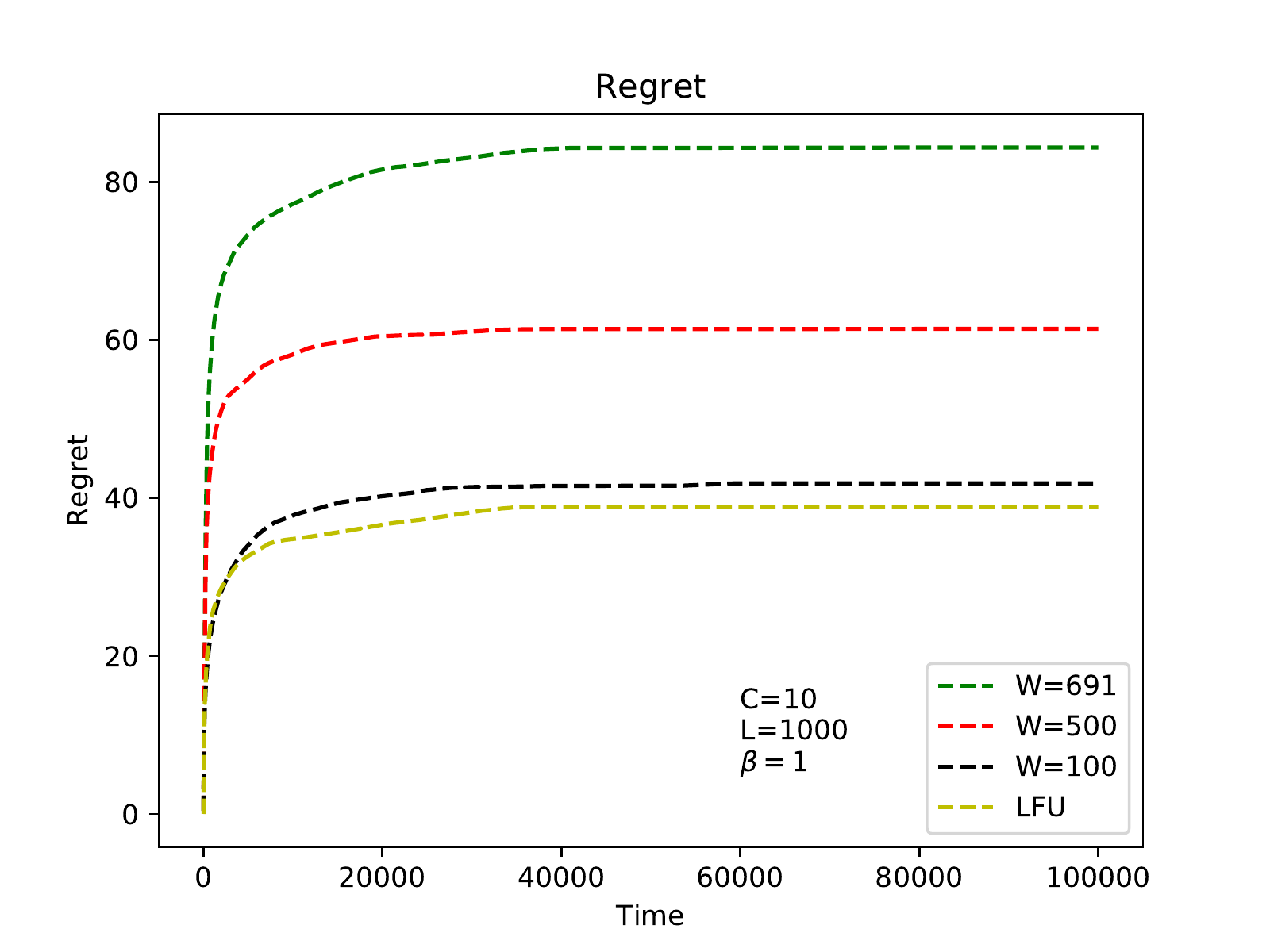}
\caption{Regret of LFU-Lite for varying W}
\label{fig:regret_w}
\end{minipage}\hfill
\begin{minipage}{.32\textwidth}
\centering
\includegraphics[width=1\columnwidth]{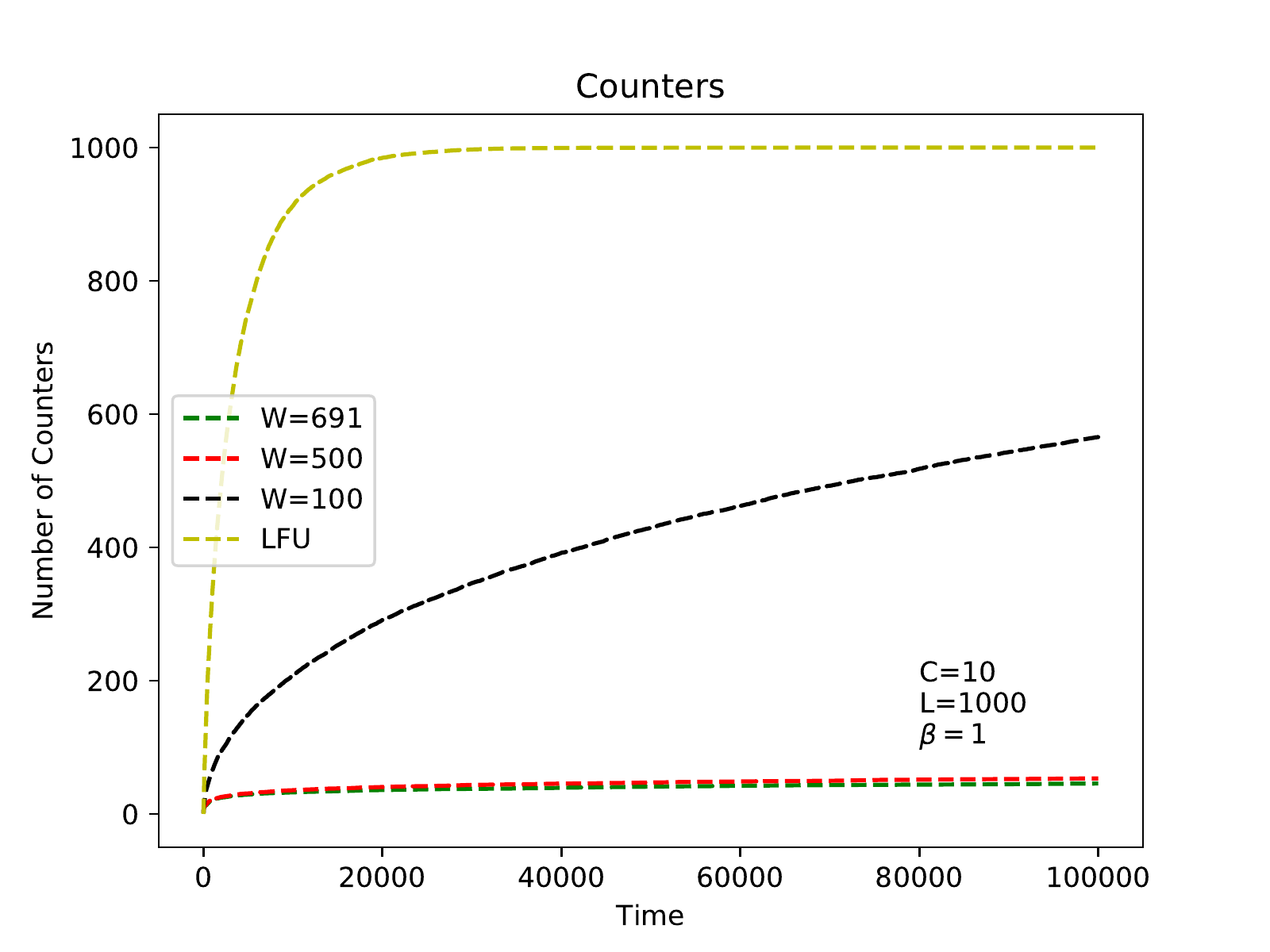}
\caption{Growth of Counters for varying W}
\label{fig:counter_w}
\end{minipage}\hfill
\centering
\begin{minipage}{.32\textwidth}
\centering
\includegraphics[width=1\columnwidth]{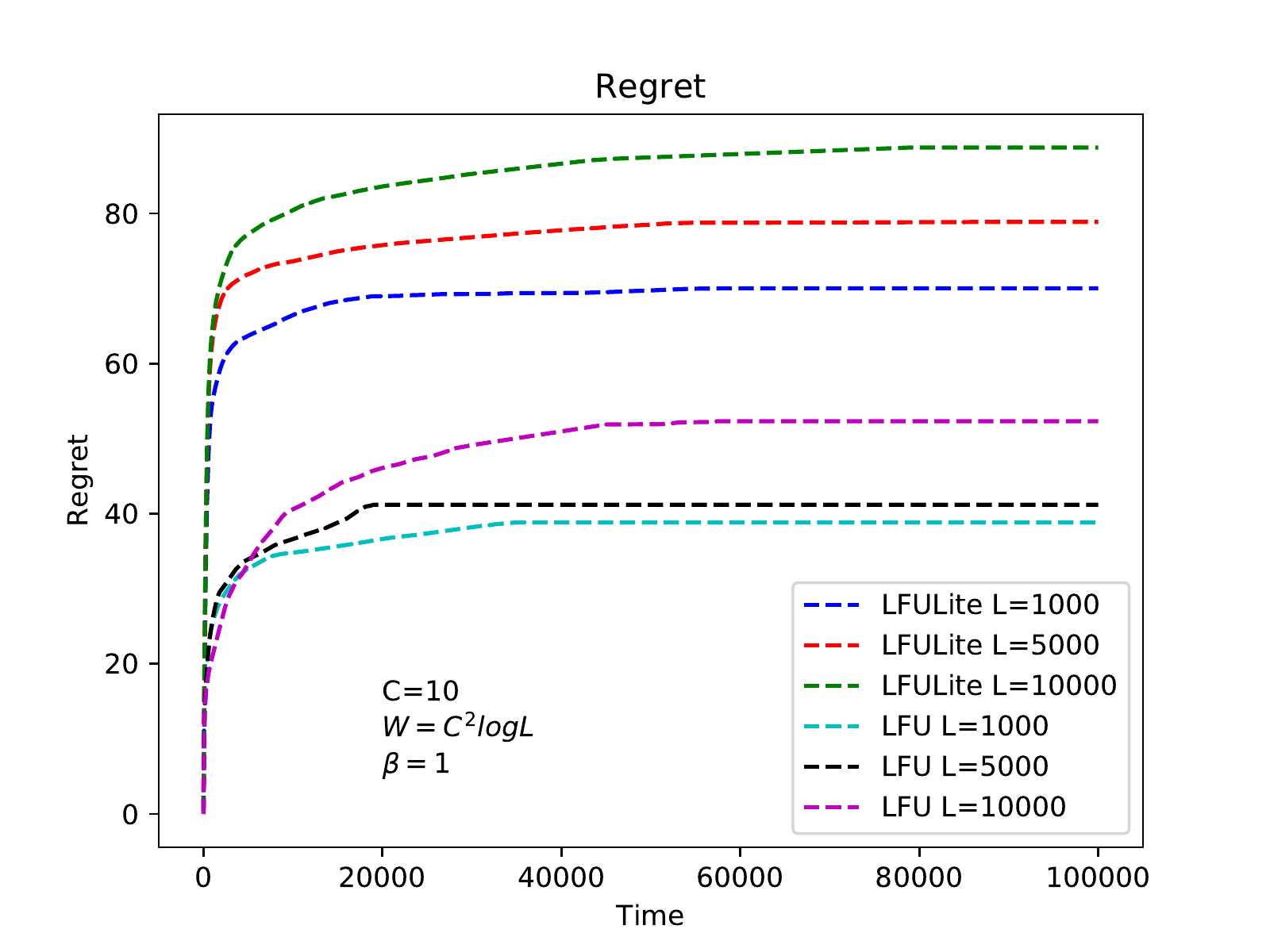}
\caption{Regret of LFU-Lite for varying L}
\label{fig:regret_l}
\end{minipage}\hfill
\begin{minipage}{.32\textwidth}
\centering
\includegraphics[width=1\columnwidth]{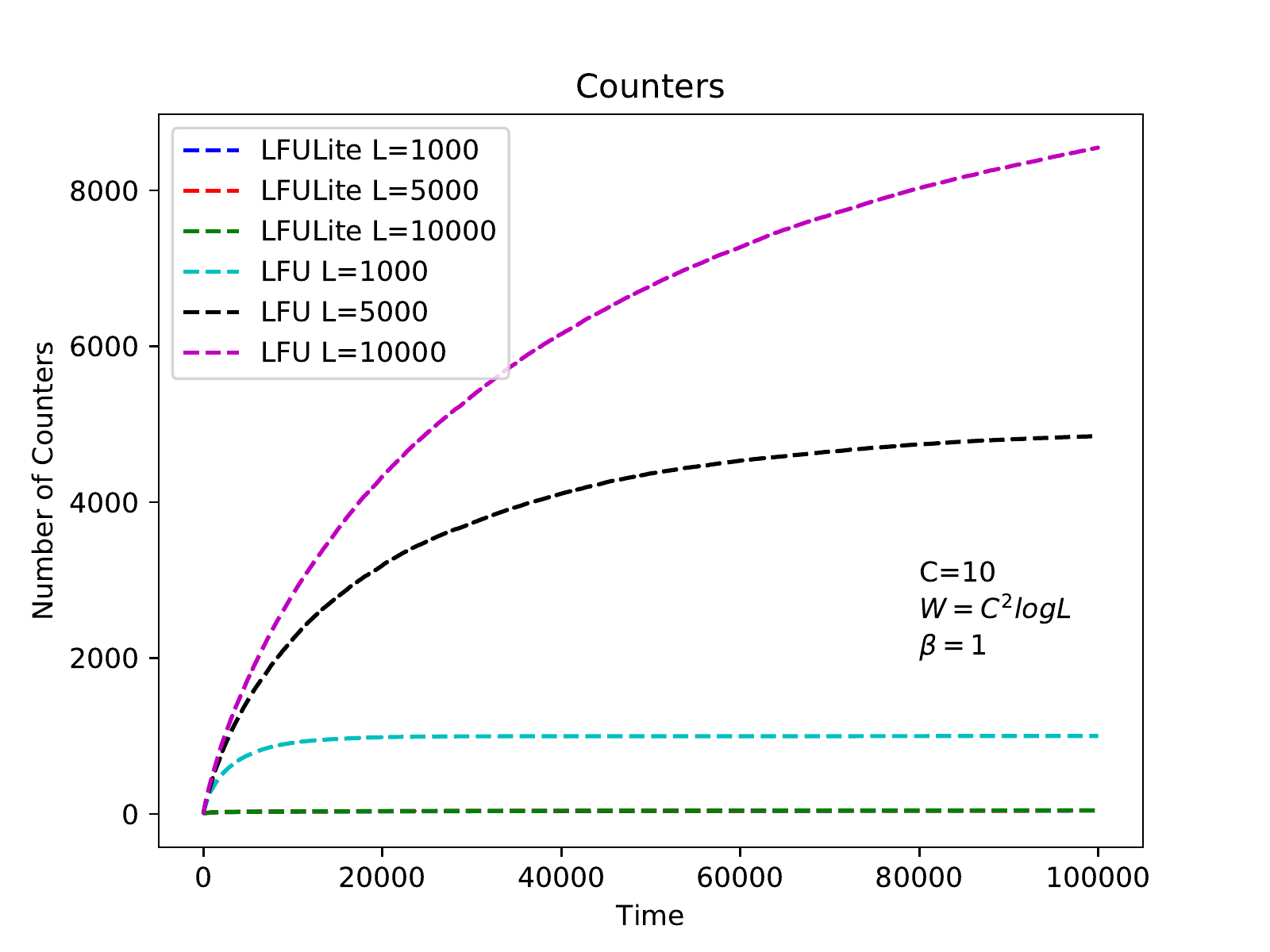}
\caption{Growth of Counters for varying L}
\label{fig:counter_l}
\end{minipage}\hfill

\centering
\begin{minipage}{.32\textwidth}
\centering
\includegraphics[width=1\columnwidth]{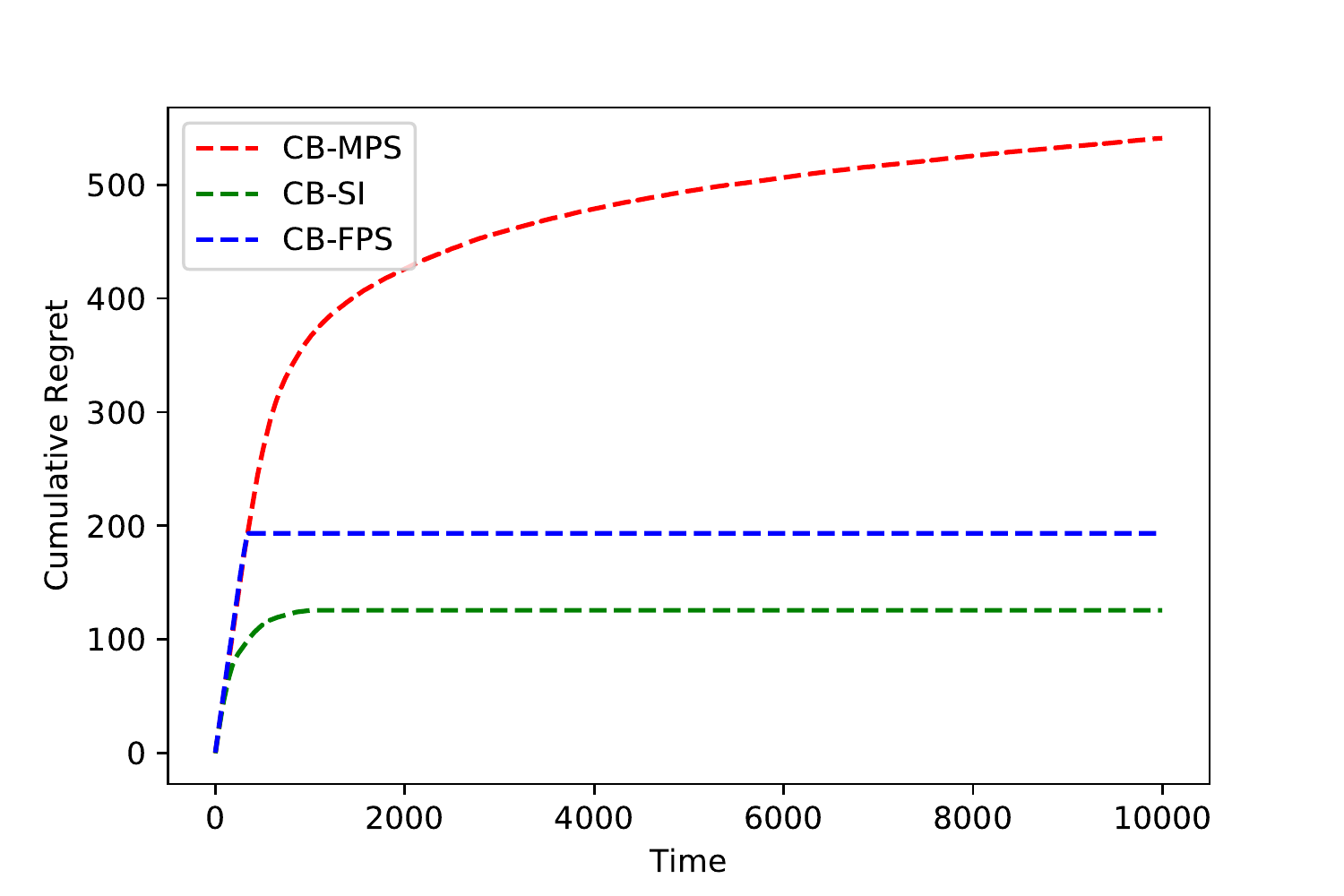}
\caption{Cumulative Regret performance comparison for $L=100,C=1,\beta=2$}
\label{fig:allalgorithms}
\end{minipage}\hfill
\begin{minipage}{.32\textwidth}
\centering
\includegraphics[width=1\columnwidth]{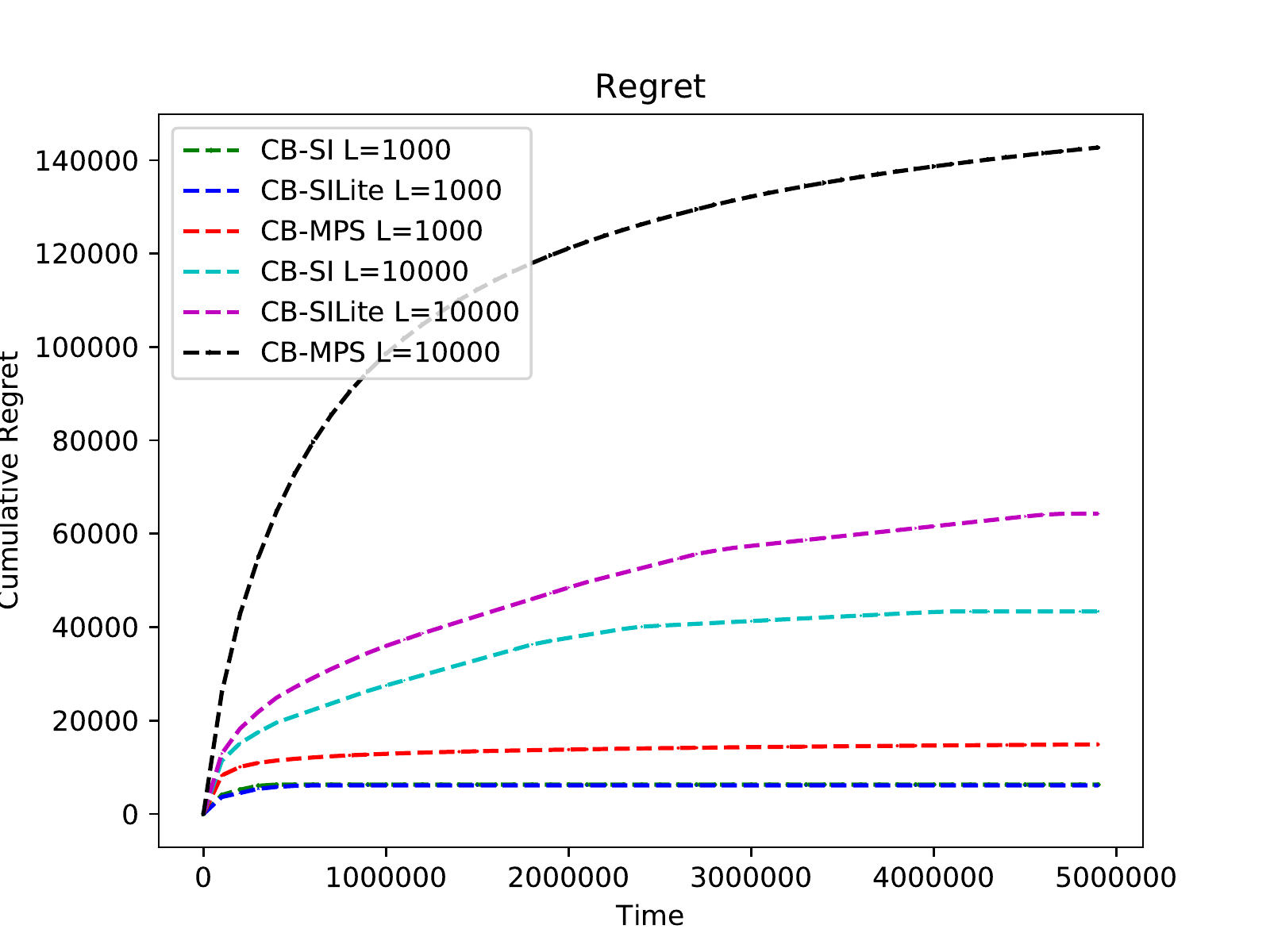}
\caption{Regret of CB-SI, CB-SILite, CB-MPS for $C=10,\beta=1$}
\label{fig:comparison}
\end{minipage}\hfill
\centering
\begin{minipage}{.32\textwidth}
\centering
\includegraphics[width=1\columnwidth]{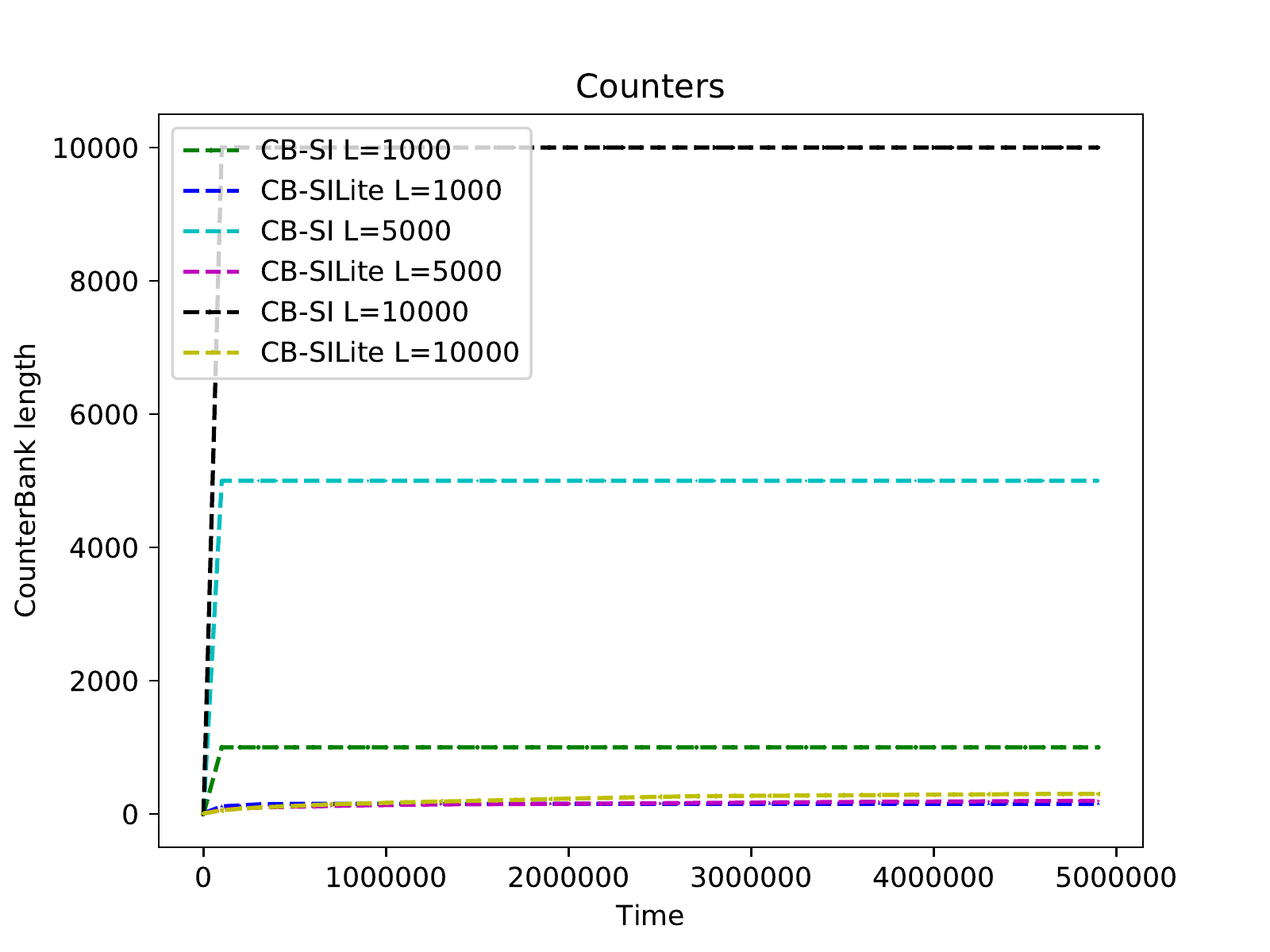}
\caption{Growth of Counters for varying L}
\label{fig:ThresholdLibraryscaling}
\end{minipage}\hfill
\end{figure*}

\begin{figure*}[htbp]
\centering
\begin{minipage}{.32\textwidth}
\centering
\includegraphics[width=1\columnwidth]{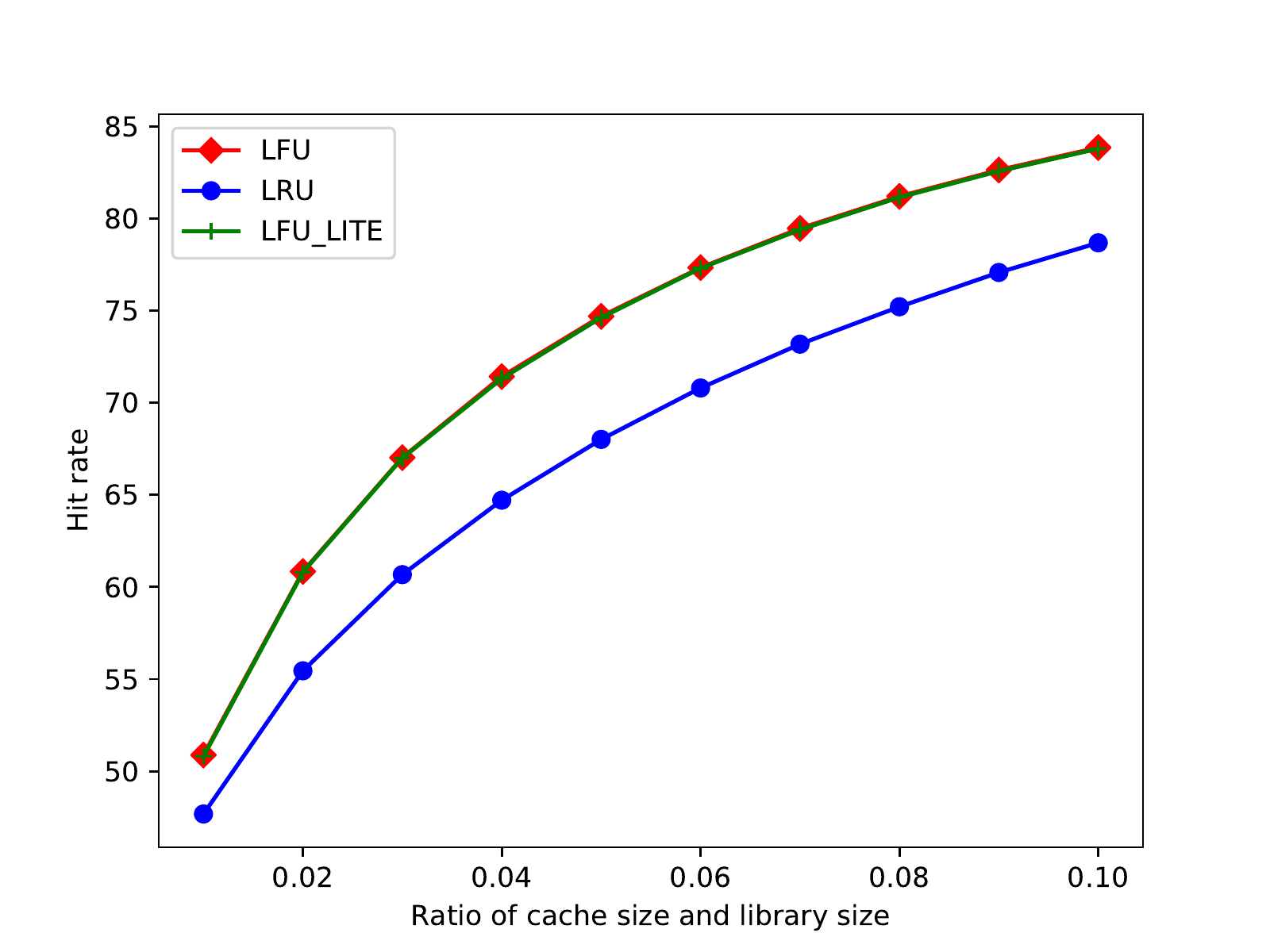}
\caption{Hit rates for Full observation for IBM trace}
\label{fig:hit_ibm}
\end{minipage}\hfill
\begin{minipage}{.32\textwidth}
\centering
\includegraphics[width=1\columnwidth]{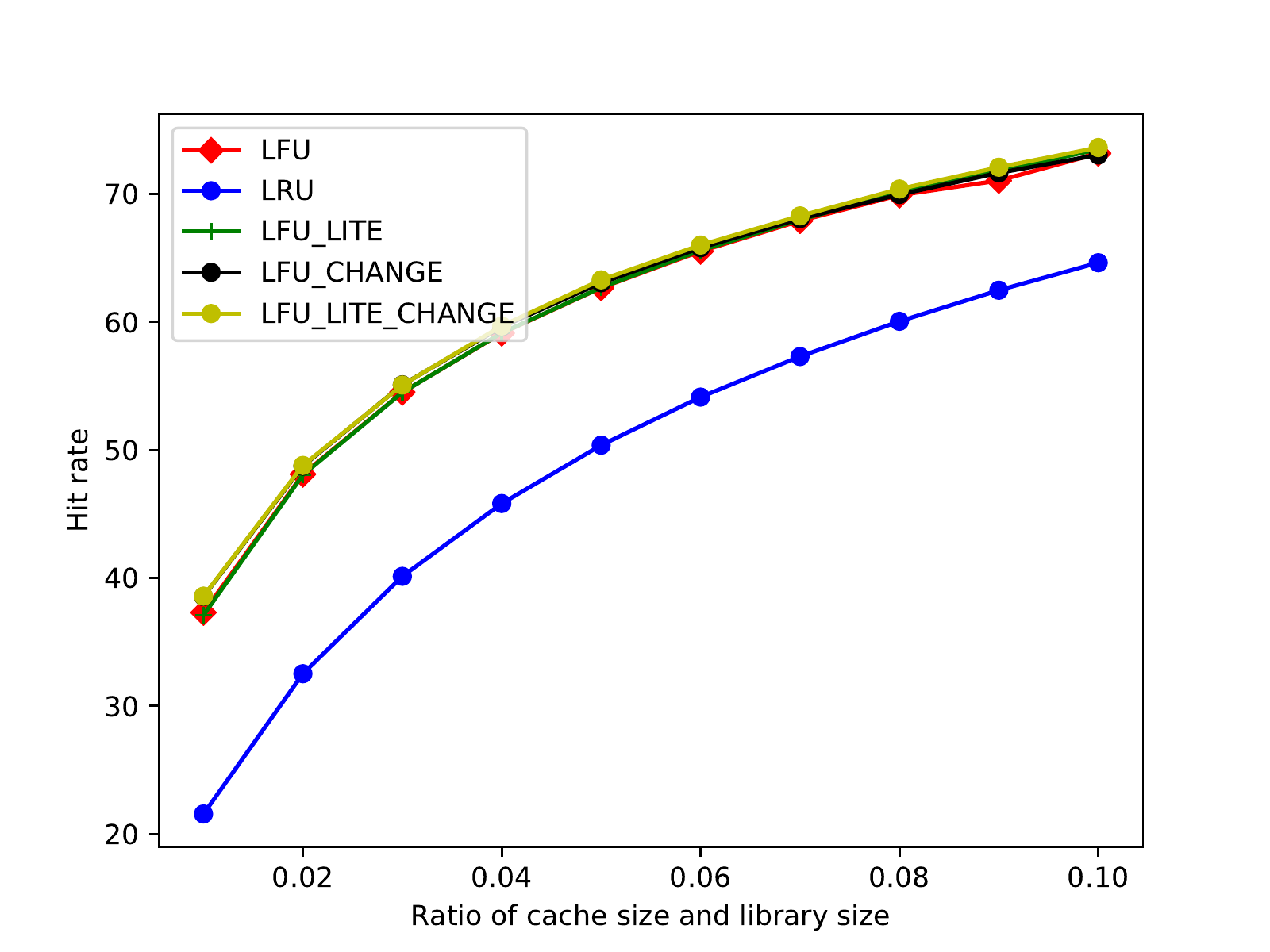}
\caption{Hit rates for Full observation for YouTube trace}
\label{fig:hit_yt}
\end{minipage}\hfill
\begin{minipage}{.32\textwidth}
\centering
\includegraphics[width=1\columnwidth]{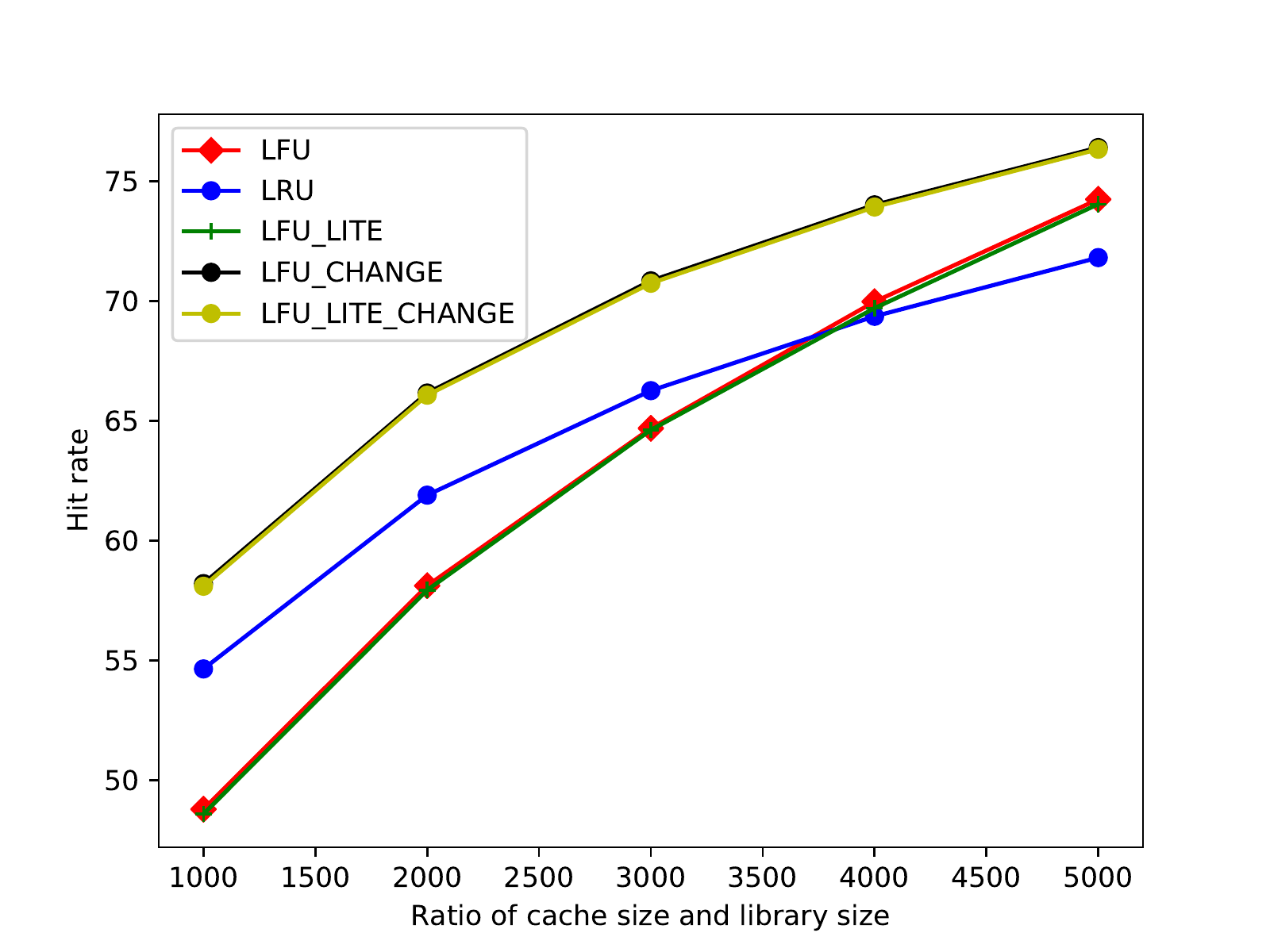}
\caption{Hit rates for Full observation for Change trace}
\label{fig:hit_ch}
\end{minipage}\hfill

\begin{minipage}{.32\textwidth}
\centering
\includegraphics[width=1\columnwidth]{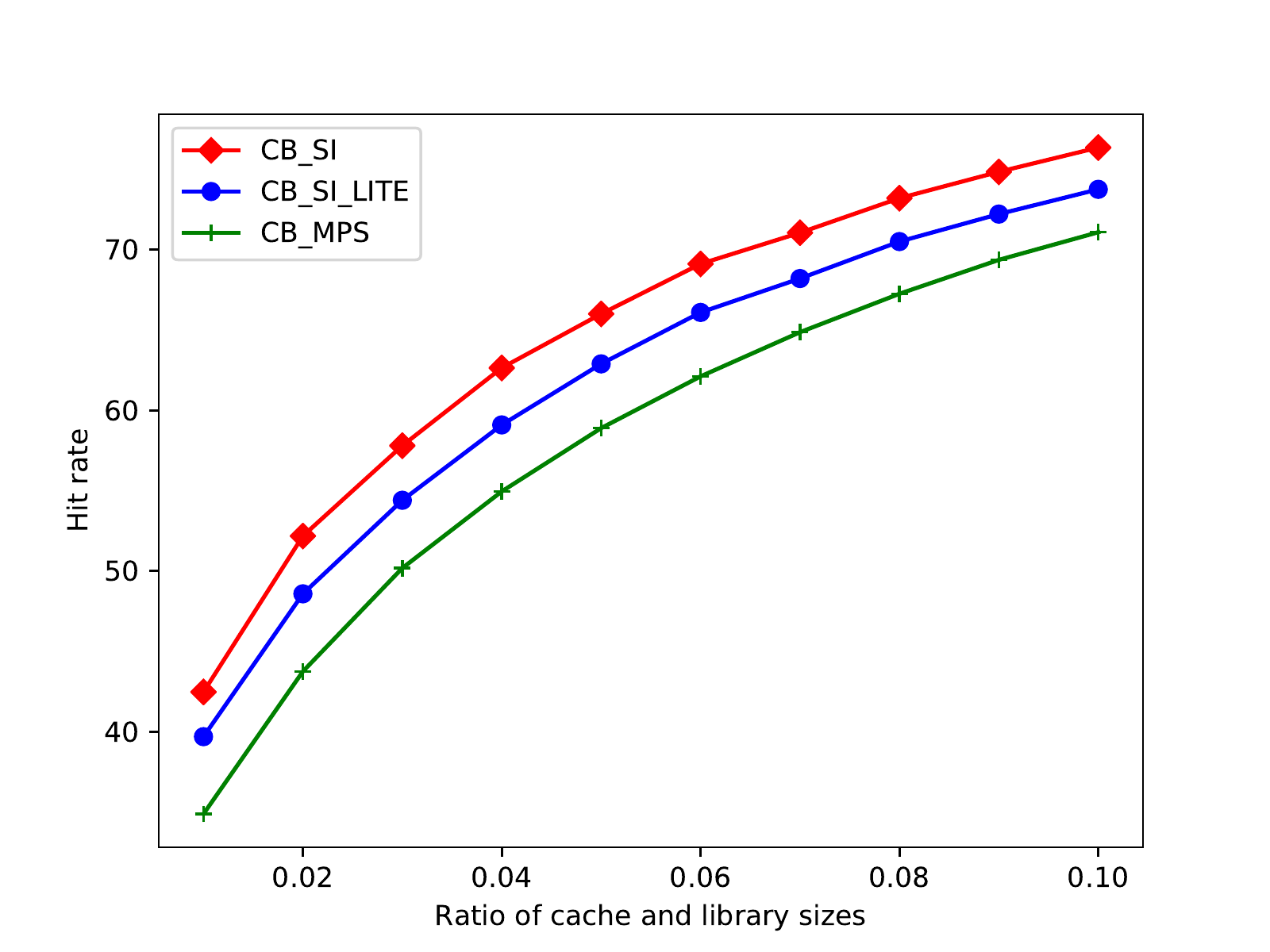}
\caption{Hit rates for partial observation, IBM trace}
\label{fig:hit_ibm_p}
\end{minipage}\hfill
\begin{minipage}{.32\textwidth}
\centering
\includegraphics[width=1\columnwidth]{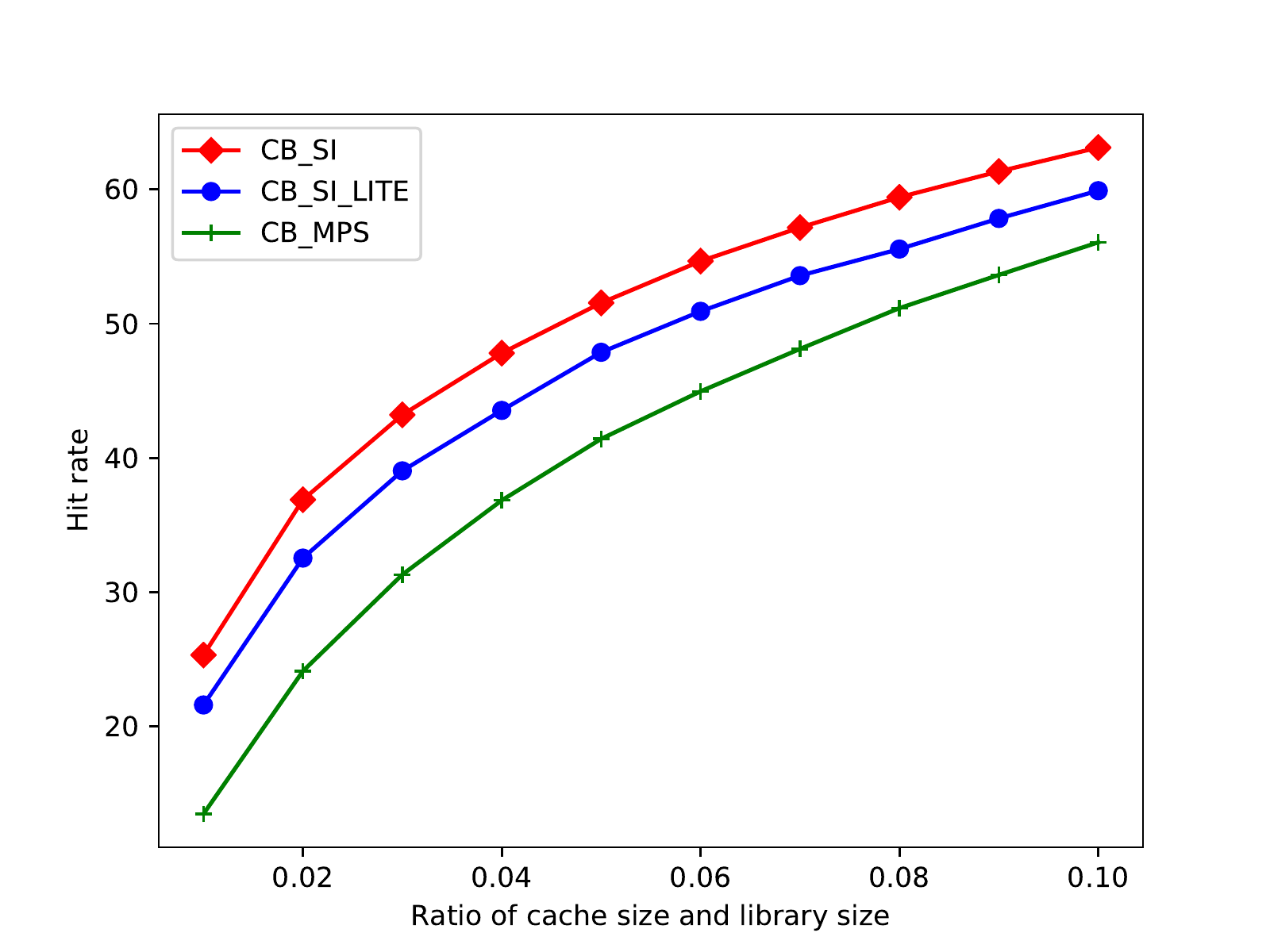}
\caption{Hit rates for Partial observation, YouTube trace}
\label{fig:hit_yt_p}
\end{minipage}\hfill
\begin{minipage}{.32\textwidth}
\centering
\includegraphics[width=1\columnwidth]{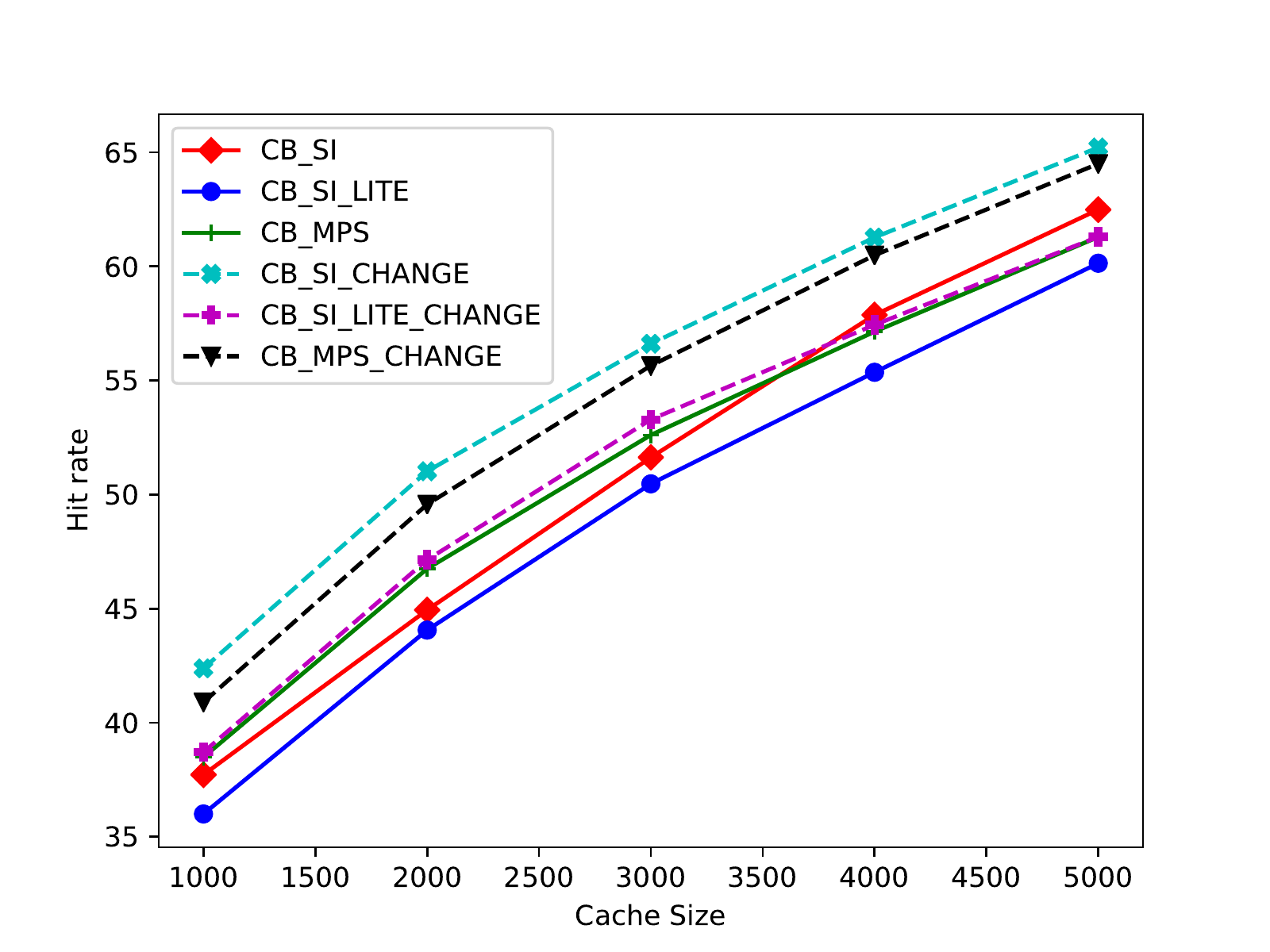}
\caption{Hit rates for partial observation, Change trace}
\label{fig:hit_change_p}
\end{minipage}\hfill

\end{figure*}

\subsection{IRM Simulations}

\subsubsection{Full Observation}
We first conduct simulations for an IRM request process following a Zipf distribution with parameter $\beta$ under the full observation setting.  Figure \ref{fig:regret_compare} compares the regret suffered by LFU, WLFU and LFU-Lite for $C=10, L=1000,W=C^2 \log L \text{ \cite{karakostas2002exploitation}}  \text{ and }\beta=1$ . As expected, the regret suffered by the WLFU algorithm grows linearly with time, while LFU and LFU-Lite suffer a constant regret.  Figure \ref{fig:counter_compare} shows the growth of the number counters used to keep an estimate of files. The merits of LFU-Lite are clearly seen here, as it uses approximately $35$ counters to achieve a constant regret, while $LFU$ uses $1000$.
 
The growth of counters and the regret suffered by LFU-Lite depends on
$W$, the window of observation. In Figures \ref{fig:regret_w} and \ref{fig:counter_w}, we compare the growth of regret and counters with $W$ for $L=1000, C=10, \beta=1$.  We see that the number of counters is essentially unchanged for a wide range of $W,$ indicating a robustness to windowing as long as it is sufficiently large.

The key advantage of the LFU-Lite algorithm is that it suffers a constant regret, while keeping track of fewer items even for large library sizes. This is clearly seen in Figure \ref{fig:regret_l} and \ref{fig:counter_l}. LFU-Lite only keeps track of approximately $45$ items while the LFU algorithm keeps track of almost all items.

\subsubsection{Partial Observation}


%

Figure \ref{fig:allalgorithms} shows the cumulative regret performance for CB-SI, CB-MPS and CB-FPS algorithms. We see that both FPS and SI versions have constant regret, whereas the MPS approach has increasing regret consistent with our analysis.  While we are forced to keep $L$ and $C$ low in Figure \ref{fig:allalgorithms} due to the complexity of the FPS approach, Figure \ref{fig:comparison} shows the cumulative regret performance for $L=1000,10000$, for CB-MPS, CB-SI and CB-SILite algorithms. As expected, the SI approach has constant regret. CB-SILite also suffers constant regret, while being a little worse compared to CB-SI. CB-MPS algorithm has logarithmic regret.

Finally, Figure \ref{fig:ThresholdLibraryscaling} shows the number of counters used by CB-SI and CB-SILite algorithms for $L=1000, 5000, 10000$, with $C=10,\beta = 1$.  The number of counters used by CB-SILite is very less even for large library sizes, compared to CB-SI.

\subsection{Trace-Based Simulations}
\label{subsec:traces}
We next conduct trace-based simulations using real world data.  The description of the traces that we use in this work is given below.
\begin{enumerate}
    \item \textbf{IBM trace:} This trace is obtained from~\cite{zerfos2013platform}. It contains a total of $1$ million requests to an IBM web server for a library of size $43857$.
    \item \textbf{YouTube trace:} This trace is obtained from~\cite{cheng2008statistics}. It contains information about the requests made for $161085$ newly created YouTube videos each week over 20 weeks.   
    From the data, we compute the popularity distribution of the videos for each week, and obtain $50000$ samples from each week's distribution.  IN this manner, we create an access trace in which the request distribution changes over each set of $50000$ requests. We run this trace for $1$ million requests. 
    \item \textbf{Synthetic trace for changing popularity distribution:}  We observe that the content popularities in the real world traces change quite slowly. Our goal is also to understand the impact of non-stationarity in the request arrival process on the hit rate performance of the proposed algorithms. To obtain a reasonable amount of non-stationarity in the popularity of items, we generate a synthetic access trace that changes the popularities periodically. To create this trace, we use a Zipf distribution with parameter $1$ to sample $1$ million requests in the following manner. For every $100000$ requests, we swap the probabilities of top $10000$ items in the access distribution cyclically, in steps of $500$. This approach results in considerable change in the distribution of the request arrivals for the top $10000$ items in the library.
\end{enumerate}
\begin{table}[h]
\begin{tabular}{|c|c|c|c|}
\hline
Cache Size & IBM           & Youtube          & Synthetic \\ \hline
$2\%$      & $15.07 \%$    & $13.25 \%$       & $26.66 \%$         \\ \hline
$4\%$      & $17.5 \% $    & $23.14 \%$       & $38.25 \%$         \\ \hline
$6\%$      & $26.56 \%$    & $30.40 \%$       & $49.42 \%$         \\ \hline
$8\%$      & $29.26 \%$    & $36.33 \%$       & $55.78 \%$         \\ \hline
$10\%$     & $31.84 \%$    & $42.52 \%$       & $62.28 \%$         \\ \hline
\end{tabular}
\caption{Counter Bank size for LFULite, Full Observation}
\label{tab:FullInfoCB}
\end{table}

\begin{table}[h]
\begin{tabular}{|c|c|c|c|}
\hline
Cache Size & IBM           & Youtube          & Synthetic \\ \hline
$2\%$      & $8.39 \%$    & $8.5 \%$       & $8.4 \%$         \\ \hline
$4\%$      & $13.24 \% $    & $13.6 \%$       & $14.09 \%$         \\ \hline
$6\%$      & $16.51 \%$    & $17.6 \%$       & $18.22 \%$         \\ \hline
$8\%$      & $19.39 \%$    & $21.23 \%$       & $22.12 \%$         \\ \hline
$10\%$     & $22.21 \%$    & $23.85 \%$       & $25.88 \%$         \\ \hline
\end{tabular}
\caption{Counter Bank size for CB-SILite, Partial Observation}
\label{tab:PartialInfoCB}
\end{table}
\subsubsection{Full Observation}
We compare the hit rates of LRU, LFU, and LFULite algorithms on the three traces.  The size of the window for LFULite is chosen  $O(C \log L),$ following the suggestions in \cite{karakostas2002exploitation}.  Figure \ref{fig:hit_ibm} shows the hit rates of LRU, LFU and LFULite on the IBM trace. We observe that LFU and LFULite outperform  LRU. Moreover, LFULite gives the same performance as LFU, while using only a fraction of counters. (Table \ref{tab:FullInfoCB}) 

For the YouTube trace (Figure \ref{fig:hit_yt}), in addition to the three algorithms, we implement heuristic versions of LFU and LFULite that account for the change in distribution. We halve the counts of LFU and LFULite every $50000$ requests. We observe that LFU and LFULite outperform  LRU, while the change versions do slightly better. The small performance gain in LFUCHANGE and LFULiteCHANGE is due to the slowly varying popularities in the YouTube trace. 

Next, we show the performance of all the five algorithms on the synthetic change trace (Figure \ref{fig:hit_ch}). We  observe that LRU dominates LFU and LFULite for small cache sizes, while LFU and LFULite outperforms LRU as the cache size grows. We also observe that the heuristic versions of LFU and LFULite outperforms LRU for all cache sizes.

\subsubsection{Partial Observation}
Figure \ref{fig:hit_ibm_p} shows the hit rates for CB-SI,CB-SILite, and CB-MPS algorithms for the IBM trace. We observe that the CB-SI algorithm clearly outperforms CB-SILite and CB-MPS algorithms. Figure \ref{fig:hit_yt_p} shows the hit rate performance of these algorithms for the YouTube trace. Even here, the performances of CB-SI and CB-SILite are superior to the CB-MPS algorithm. 

For the synthetic change trace, we also implement change versions of all the three algorithms by halving the counts periodically every $50000$ requests. We notice that CB-SIChange outperforms all the other algorithms.  We observe that CB-SILite uses small counter bank for all the traces as shown in see Table \ref{tab:PartialInfoCB}.

\section{Conclusion}

We considered the question of caching algorithm design and analysis from the perspective of online learning.  We focused on algorithms that estimate popularity by maintaining counts of requests seen, in both the full and partial observation regimes.  Our main findings were in the context of full observation, it is possible to follow this approach and obtain $O(1)$ regret using the simple LFU-Lite approach that only needs a small number of counters.  In the context of partial observations, our finding using the CB-SI approach was that structure greatly enhances the learning ability of the caching algorithm, and is able to make up for incomplete observations to yield $O(1)$ regret.  We verified these insights using both simulations and data traces.  In particular, we showed that even if the request distribution changes with time, our approach (enhanced with a simple "forgetting" rule) is able to outperform established algorithms such as LRU.

\bibliographystyle{IEEEtran} 
\bibliography{refs}

\end{document}